\documentclass{siamonline190516}

\usepackage{amssymb, amsmath}
\usepackage{algorithmic}
\usepackage{cleveref}
\usepackage{graphicx}
\usepackage[caption=false]{subfig}
\usepackage{siunitx, tabularx, booktabs, enumitem}
\newcolumntype{Y}{>{\centering\arraybackslash}X}

\usepackage{enumitem}
\setlist[enumerate]{leftmargin=.5in}
\setlist[itemize]{leftmargin=.5in}

\newsiamremark{note}{Note}
\crefname{note}{Note}{Notes}

\headers{Multifidelity ABC-SMC}{T.P. Prescott and R.E. Baker}

\title{Multifidelity Approximate Bayesian Computation with Sequential Monte Carlo Parameter Sampling\thanks{
		Originally submitted to the editors January 2020; resubmitted October 2020. \funding{
			This work is supported by BBSRC through grant BB/R000816/1.}}}

\author{
	Thomas P. Prescott\thanks{Mathematical Institute, University of Oxford, Woodstock Road, Oxford, OX2 6GG, UK (\email{prescott@maths.ox.ac.uk, ruth.baker@maths.ox.ac.uk})\funding{REB is supported by a Royal Society Wolfson Research Merit
			Award and a Leverhulme Research Fellowship.}}
	\and Ruth E. Baker\footnotemark[2]
}

\newcommand{\obs}[1]{#1_{\mathrm{obs}}}

\newcommand{\fp}[1]{#1_{\mathrm{fp}}}
\newcommand{\fn}[1]{#1_{\mathrm{fn}}}

\allowdisplaybreaks
\usepackage{setspace}

\begin{document}

\maketitle{}

\begin{abstract}
Multifidelity approximate Bayesian computation (MF-ABC) is a likelihood-free technique for parameter inference that exploits model approximations to significantly increase the speed of ABC algorithms~(Prescott and Baker, JUQ, 2020).
Previous work has considered MF-ABC only in the context of rejection sampling, which does not explore parameter space particularly efficiently.
In this work, we integrate the multifidelity approach with the ABC sequential Monte Carlo (ABC-SMC) algorithm into a new MF-ABC-SMC algorithm.
We show that the improvements generated by each of ABC-SMC and MF-ABC to the efficiency of generating Monte Carlo samples and estimates from the ABC posterior are amplified when the two techniques are used together.
\end{abstract}



\doublespacing
\section{Introduction}
\label{s:Intro}

An important goal of the mathematical modelling of a physical system is to be able to make quantitative predictions about its behaviour.
In order to make accurate predictions, the parameters of the mathematical model need to be calibrated against experimental data.
Bayesian inference is a widely-used approach to model calibration that seeks to unify the information from observations with prior knowledge about the parameters to form a posterior distribution on parameter space~\cite{Beaumont2010,Hines2015,Schnoerr2017}.
This approach is based on Bayes' Theorem, whereby the posterior distribution is proportional to the product of the prior distribution and the likelihood of the data under the model.

However, in many practical settings, the model is often too complicated for the likelihood of the data to be calculated, making the posterior distribution unavailable.
In this case, likelihood-free methods for Bayesian parameter inference become a useful option.
Specifically, approximate Bayesian computation (ABC) is a class of such likelihood-free methods~\cite{Sisson2018,Sunnaker2013}.
Rather than calculating the likelihood of the observed data for a given parameter value, it is estimated. 
The model is simulated under a particular parameter value, and the simulated output is then compared with the observed data.
If the simulation and observations are suitably close (according to a predetermined metric and threshold~\cite{Fearnhead2012,Harrison2017a}) then, in the classical rejection sampling approach, the likelihood is approximated as $1$.
Otherwise, the likelihood is approximated as $0$.
This binary approximation is usually interpreted as an acceptance or a rejection of the parameter value input into the simulation.
Using this approach, a weighted Monte Carlo sample from an approximated posterior can be built by repeatedly proposing parameters from the prior distribution, simulating the model with each parameter, and calculating the binary weight.

One widely-acknowledged weakness of the ABC approach is its heavy reliance on repeated simulation.
There has been a significant amount of work dedicated to overcoming this reliance by exploiting methods for intelligent exploration of parameter space in order to reduce the number of simulations in areas of low likelihood~\cite{Sisson2018}.
For example, the parameters to be input into the simulation can instead be proposed from a importance distribution, rather than the prior, which we aim to construct in order to improve the algorithm's performance.
One successful approach to importance sampling is known as Sequential Monte Carlo ABC (ABC-SMC), which aims to build consecutive samples using parameter proposals taken from progressively closer approximations to the posterior, parameterised by decreasing ABC thresholds~\cite{Toni2009}.
Research in this area has considered how to choose the sequence of decreasing thresholds and distance metrics~\cite{DelMoral2012,Prangle2017}, and how best to evolve the parameters from one approximation to produce parameter proposals for the next~\cite{Alsing2018,Beaumont2009,Filippi2013}.

Another proposed strategy for overcoming the simulation bottleneck is the multifidelity ABC (MF-ABC) approach~\cite{Prescott2020}.
This approach assumes that, in addition to the model under investigation (termed the \emph{high-fidelity} model), there also exists a \emph{low-fidelity} model, depending on the same parameters, that is significantly faster to simulate, usually at the cost of being less accurate~\cite{Peherstorfer2018}.
The multifidelity approach to parameter inference uses simulations from both the low-fidelity model and high-fidelity model to approximate the likelihood.
The high-fidelity model is simulated as little as possible, to reduce the simulation bottleneck, but just enough to ensure that the resulting posterior estimate is suitably accurate.
This technique is related to multilevel approaches to ABC~\cite{Jasra2017,Tran2015,Warne2018}, which use a decreasing sequence of ABC-SMC thresholds to produce coupled estimates at different levels.
Other approaches that can be interpreted in the multifidelity framework include lazy ABC~\cite{Prangle2016}, delayed acceptance ABC~\cite{Everitt2017} and early rejection~\cite{Picchini2016}.
In each of these, low-fidelity simulations are sometimes used to reject parameters before completing a high-fidelity simulation.
Importantly, the more general MF-ABC framework in \cite{Prescott2020} allows for early acceptance as well as early rejection.

A key observation that can be made about these two techniques for improving ABC performance is that they are orthogonal, in the sense that they improve different aspects of the ABC approach.
ABC-SMC considers only improving the method for proposing parameters to use in simulations and does not directly affect the binary estimate of the likelihood.
In contrast, MF-ABC makes no change to the parameter proposals, but instead directly alters the method used to estimate the likelihood by using a combination of both low-fidelity and high-fidelity model simulations.
The complementarity of these two approaches has previously been shown in the specific context of combining delayed acceptance with SMC~\cite{Everitt2017}.
Thus, combining the general multifidelity framework of~\cite{Prescott2020} with SMC should therefore yield significant speed-up over existing methods.

\subsection{Outline}
\label{s:Outline}
In this paper we bring together these two orthogonal approaches to speeding up ABC algorithms.
We will introduce a combined multifidelity sequential Monte Carlo ABC algorithm (MF-ABC-SMC).
\Cref{s:Background} formulates the existing ABC algorithms briefly described above, and the techniques we can use to quantify their performance.
We then show how these ABC approaches can be combined in \Cref{s:MF-ABC-SMC}, by incorporating the multifidelity technique with the sequential importance sampling approach to ABC-SMC.
In \Cref{MFABCSMC} we then fully exploit the SMC framework to optimise the multifidelity approach, producing the MF-ABC-SMC algorithm in \Cref{MFABC:SMC}.
This new algorithm is applied in \Cref{s:Example} to a heterogeneous network of Kuramoto oscillators in a hierarchical Bayes parameter estimation task, to produce low-variance ABC posterior estimates significantly faster than the classical ABC-SMC approach.
Finally, in \Cref{s:Discussion}, we discuss some important open questions for further optimising the MF-ABC-SMC algorithm.

\section{Theoretical background}
\label{s:Background}

Assume that the model we are seeking to calibrate is a map (usually stochastic) from parameters $\theta$, taking values in a parameter space $\Theta$, to an output $y$, taking values in data space $\mathcal Y$.
We denote this map as a conditional density $f(\cdot~|~\theta)$ on $\mathcal Y$, and term the drawing of $y \sim f(\cdot~|~\theta)$ as simulating the model, with $y$ termed a simulation.
For Bayesian inference, we furthermore assume the existence of a prior distribution $\pi(\cdot)$ on $\Theta$, and of the observed data $\obs y \in \mathcal Y$ that will be used to calibrate the model.
The model induces the likelihood of the observed data, written $L(\theta) = f(\obs y~|~\theta)$, which is a function of $\theta$.
As described previously, the goal of Bayesian inference is to infer the posterior distribution $p(\theta~|~\obs y)$ on $\Theta$, given $\obs y$ and $\pi(\cdot)$.
Bayes' Theorem equates the posterior to the product of likelihood and prior,
\[
p( \theta ~|~ \obs y) = \frac{1}{\zeta} L(\theta) \pi(\theta),
\]
where the normalisation constant, $\zeta$, ensures that $p(\theta~|~\obs y)$ is a probability density that integrates to unity.

\subsection{Approximate Bayesian computation}
\label{s:ABC}

Often, the model under consideration is sufficiently complicated that the likelihood cannot be calculated, necessitating a likelihood-free approach.
We assume that, while the value of any $f(y ~|~ \theta)$ is not available, we are still able to simulate $y \sim f(\cdot ~|~ \theta)$.
Let $d(y, \obs y)$ denote a metric that quantifies how close any simulation, $y$, is to the observed data, $\obs y$. 
For a positive threshold value $\epsilon > 0$, we can then define a neighbourhood $\Omega_\epsilon(d, \obs y) = \{ y \in \mathcal Y ~|~ d(y, \obs y)<\epsilon \}$ of model simulations that are `close' to the data.
Typically the dataset, $\obs y$, is constant for the parameter estimation task, and the distance metric, $d$, is pre-determined.
Hence we will often drop the $(d, \obs y)$ dependence from our notation and simply write $\Omega_\epsilon$ for the $\epsilon$-neighbourhood of $\obs y$ under the distance function $d$.

For a given positive distance threshold $\epsilon > 0$, ABC replaces the exact likelihood, $L(\theta)$, with the ABC approximation to the likelihood,
\begin{subequations}
\label{eq:ABC}
\begin{equation}
\label{eq:ABClikelihood}
 L_\epsilon(\theta) = \mathbb P(y \in \Omega_\epsilon ~|~ \theta) = \int \mathbb I(y \in \Omega_\epsilon) f(y~|~\theta) ~\mathrm dy,
\end{equation}
which is, to leading order for small $\epsilon$, approximately proportional to $L(\theta)$.
The ABC approximation to the likelihood then induces the ABC posterior,
\begin{equation}
\label{eq:ABCposterior}
 p_{\epsilon}(\theta~|~\obs y) = \frac{1}{Z} ~ L_\epsilon(\theta) \pi(\theta),
\end{equation}
\end{subequations}
where (similarly to $\zeta$ above) the constant $Z$ ensures that the ABC posterior is a probability distribution with unit integral.

\begin{algorithm}
\caption{Importance sampling ABC (ABC-IS)}
\label{ABC:Importance}
\begin{algorithmic}[1]
\REQUIRE{
Data $\obs y$ and neighbourhood $\Omega_\epsilon$; 
model $f(\cdot~|~\theta)$; 
prior $\pi$; 
importance distribution $\hat q$ proportional to $q$; 
sample index $n=0$; 
stopping criterion $S$.
}
\ENSURE{Weighted sample $\left\{ \theta_n, w_n \right\}_{n=1}^{N}$.}
\REPEAT{}
 \STATE{Increment $n \leftarrow n+1$.}
 \STATE{Generate $\theta_n \sim \hat q(\cdot)$.}
 \STATE{Simulate $y_n \sim f(\cdot~|~\theta_n)$.}
 \STATE{Set $w_n = \left[ \pi(\theta_n) \big/ q(\theta_n) \right] \cdot \mathbb I \left( y_n \in \Omega_\epsilon \right) $.}
\UNTIL{$S =$ \texttt{true}.}
\end{algorithmic}
\end{algorithm}

The importance sampling ABC algorithm~\cite{Owen2013,Sisson2018} presented in \Cref{ABC:Importance} (ABC-IS) presents a simple method for drawing samples from the ABC posterior, $p_\epsilon(\theta~|~\obs y)$.
In addition to the data, $\obs y$, and the prior, $\pi(\theta)$, we assume an importance probability distribution $\hat q(\theta) = q(\theta) / Z_q$ defined by the function $q(\theta)$, where $q(\theta)>0$ for all $\theta$ in the support of $\pi$.
Note that we do not assume that the normalisation constant $Z_q$ is known.
For each parameter proposal $\theta_n \sim \hat q(\cdot)$ from the importance distribution, the model $y_n \sim f(\cdot~|~\theta_n)$ is simulated and a weight $w_n = w(\theta_n, y_n)$ is assigned, using the weighting function,
\begin{equation}
 \label{eq:ImportanceWeight}
 w(\theta, y) = \left[ \pi(\theta) \big/ q(\theta) \right] \cdot \mathbb I(y \in \Omega_\epsilon) \geq 0,
\end{equation}
to produce a weighted sample $\{ w_n, \theta_n \}$.

Note that the general stopping criterion used in \Cref{ABC:Importance} allows for the algorithm to terminate after (for example) a fixed number of non-zero weights, a fixed number of parameter proposals, $N$, when a fixed budget of total computational time or memory is reached, or any other more complicated combination of such conditions.
Furthermore we note that ABC-IS is easily parallelised, although care must be taken to ensure that the chosen stopping condition is correctly applied in this case~\cite{Jagiella2017}.
If we choose the importance distribution $\hat q = \pi$ equal to the prior, then \Cref{ABC:Importance} (ABC-IS) is known as rejection sampling, which we refer to as ABC-RS.

For any arbitrary function $F(\cdot)$ defined on the parameter space $\Theta$, we can estimate the expected value of $F(\theta)$ under $p_\epsilon(\theta~|~\obs y)$, such that
\begin{equation}
 \label{eq:estimate}
 \mathbb E_{p_\epsilon}(F(\theta)) \approx \bar F = \frac{\sum_{n=1}^N w_n F(\theta_n)}{\sum_{n=1}^N w_n}.
\end{equation}
Although this estimate is biased (except in the ABC-RS case), it is consistent, such that the mean squared error (MSE) of $\bar F$ is dominated by the variance and decays to $0$ on the order $1/N$.

\subsection{Sequential Monte Carlo (SMC)}
\label{s:SMC}

Sequential Monte Carlo (SMC) is commonly used to efficiently explore parameter space in ABC.
The goal is to propagate a sample from the prior through a sequence of intermediate distributions towards the target distribution, $p_\epsilon(\theta~|~\obs y)$.
The intermediate distributions are typically the sequence of ABC approximations $p_{\epsilon_t}$, for $t=1,\dots,T$, defined by a sequence of decreasing thresholds, $\epsilon_1 > \dots > \epsilon_T = \epsilon$.
\Cref{ABC:SMC} presents the sequential importance sampling approach to ABC-SMC~\cite{Beaumont2009,DelMoral2006,Sisson2007,Toni2009}, also known as population Monte Carlo~\cite{Cappe2004}.
Each Monte Carlo sample $\{ \theta_n^{(t)}, w_n^{(t)} \}$ built at generation $t$ is used to construct an importance distribution $\hat q_{t+1}$, defined in \Cref{eq:importance}, that is used to generate the next generation's Monte Carlo sample.
The final sample, $\{ \theta_n^{(T)}, w_n^{(T)} \}$, is produced by \Cref{ABC:Importance} using importance distribution $\hat q_T$ and threshold $\epsilon_T = \epsilon$.
Note that \Cref{ABC:SMC} requires the specification of a threshold sequence $\epsilon_t$, the stopping conditions $S_t$, and the perturbation kernels $K_t(\cdot~|~\theta)$ for $t=1,\dots,T$~\cite{Sisson2018}.
We will not re-examine these aspects of ABC-SMC in detail in this paper, and will implement ABC-SMC using established techniques to select $\epsilon_t$, $S_t$ and $K_t$~\cite{Beaumont2009,DelMoral2012,Filippi2013,Toni2009}.

\begin{algorithm}
\caption{Sequential Monte Carlo ABC (ABC-SMC)}
\label{ABC:SMC}
\begin{algorithmic}[1]
\REQUIRE{
Data $\obs y$; 
sequence of nested neighbourhoods $\Omega_{\epsilon_T} \subseteq \Omega_{\epsilon_{T-1}} \subseteq \cdots \subseteq \Omega_{\epsilon_1}$ for $0 < \epsilon = \epsilon_T < \epsilon_{T-1} < \dots < \epsilon_1$; 
prior $\pi$; 
perturbation kernels $K_t(\cdot~|~\theta)$; 
initial importance distribution $\hat q_1$ (often set to $\pi$); 
model $f(\cdot~|~\theta)$; 
stopping conditions $S_1, S_2, \dots, S_T$.}
\ENSURE{Weighted sample $\left\{ \theta_n^{(T)}, w_n^{(T)} \right\}_{n=1}^{N_T}$.}
\FOR{$t=1, \dots, T-1$}
\STATE{Produce $\{ \theta_n^{(t)}, w_n^{(t)}\}_{n=1}^{N_t}$ using \Cref{ABC:Importance} (ABC-IS) with importance distribution $\hat q_t$, neighbourhood $\Omega_{\epsilon_t}$, and stopping condition $S_t$.}
\STATE{Define the next importance distribution, $\hat q_{t+1}$, proportional to 
\begin{equation}
\label{eq:importance}
q_{t+1}( \theta ) = \begin{cases}
                    	\sum_{n=1}^{N_t} w_n^{(t)} K_t( \theta ~|~\theta_n^{(t)}) \bigg/ \sum_{m=1}^{N_t} w_m^{(t)} & \pi(\theta)>0 , \\
                     0 &\text{else.}
                    \end{cases}
\end{equation}}
\ENDFOR{}
\STATE{Produce $\left\{ \theta_n^{(T)}, w_n^{(T)} \right\}_{n=1}^{N_T}$ using \Cref{ABC:Importance} (ABC-IS) with importance distribution $\hat q_T$, neighbourhood $\Omega_{\epsilon_T} = \Omega_\epsilon$, and stopping condition $S_T$.}
\end{algorithmic}
\end{algorithm}

One weakness of the population Monte Carlo approach taken in \Cref{ABC:SMC} is the $O(N^2)$ cost for each run of \Cref{ABC:Importance} (ABC-IS), where $N_t \sim N$ is the scale of each generation's sample size.
To overcome this problem, the SMC sampler was adapted in \cite{DelMoral2012} to reduce this cost to $O(N)$.
When the calculation of each $w_n^{(t)}$ is dominated by $q_{t}(\theta_n^{(t)})$, then a significant computational burden can be alleviated through using the $O(N)$ sampling approach.
However, in many practical problems, the calculation of each weight, $w_n^{(t)}$, is instead dominated by sampling $y_n^{(t)} \sim f(\cdot~|~\theta^{(t)}_n)$.
We will focus on the latter setting and aim to reduce the cost of ABC-SMC through reducing the cost of calculating each $w_n^{(t)}$ using the multifidelity approach described in \Cref{s:MFABC}.

\subsection{Multifidelity ABC}
\label{s:MFABC}
\Cref{s:SMC} describes how the SMC strategy provides parameter proposals such that simulation time is not wasted in regions of low likelihood.
An orthogonal approach to improving ABC efficiency is to avoid computationally expensive simulations, where possible, by relying on the \emph{multifidelity} framework~\cite{Prescott2020}.
We term the model of interest, $f(\cdot~|~\theta)$, which maps parameter space $\Theta$ to an output space $\mathcal Y$, as the high-fidelity model.
We now assume that, in addition, there is a low-fidelity (i.e. approximate) model, $\tilde f(\cdot~|~\theta)$, of the same physical system with the same parameter space $\Theta$.
The simulations of this model are denoted $\tilde y \sim \tilde f(\cdot~|~\theta)$, taking values in the output space $\tilde{\mathcal Y}$, which may differ from $\mathcal Y$.
Importantly, we assume that the low-fidelity model is computationally cheaper, in the sense that simulations $\tilde y \sim \tilde f(\cdot~|~\theta)$ of the low-fidelity model incur less computational burden than simulations $y \sim f(\cdot~|~\theta)$ of the high-fidelity model.

We can also assume that the experimental observations $\obs y \in \mathcal Y$ can be mapped to the new data space, giving $\obs{\tilde y} \in \tilde{\mathcal Y}$.
Similarly, we define an associated region,
\[
 \tilde \Omega_{\epsilon} = \tilde \Omega_{\epsilon} (\obs{\tilde y}, \tilde d) = \{ \tilde y \in \tilde{\mathcal Y} ~|~ \tilde d(\tilde y,\obs{\tilde y}) < \tilde \epsilon \},
\]
that is the $\tilde \epsilon$-neighbourhood of the observed data, $\obs{\tilde y}$, as a subset of the output space $\tilde{\mathcal Y}$, defined under the distance metric $\tilde d$.
However, in the interests of clarity, we will assume for the remainder of this article that the output spaces of each model fidelity are such that $\tilde{\mathcal Y} = \mathcal Y$.
Similarly, we assume that the observed data are such that $\obs{\tilde y} = \obs y$, the distance metrics are such that $\tilde d = d$, and the ABC thresholds are such that $\tilde \epsilon = \epsilon$, so that $\tilde \Omega_{\tilde \epsilon} = \Omega_\epsilon$.

In general, the models $f(\cdot~|~\theta)$ and $\tilde f(\cdot~|~\theta)$ can be simulated independently for a given $\theta$.
Then, if the low-fidelity model is a good approximation to the high-fidelity model, the outputs $y$ and $\tilde y$ will be near, in some sense, and the distances from data $d(y, \obs y)$ and $d(\tilde y, \obs y)$ will be correlated.
However, to improve this correlation we will also allow for coupling between the two models,
writing $\check f(y, \tilde y~|~\theta)$ as the coupled density for $(y, \tilde y)$ with marginals that coincide with the independent models $f(y~|~\theta)$ and $\tilde f(\tilde y~|~\theta)$.
The benefit of this approach is that, with a judicious choice of coupling, we can produce a high-fidelity simulation $y \sim f(\cdot~|~\tilde y, \theta)$ conditionally on a previously-simulated low-fidelity simulation, $\tilde y \sim \tilde f(\cdot~|~\theta)$,
where the distances are more closely correlated.
Furthermore, simulations of $y \sim f(\cdot~|~\tilde y, \theta)$ may be less computationally burdensome than independent simulations of $y \sim f(\cdot~|~\theta)$.

For example, suppose the high-fidelity model is a Markovian stochastic dynamical system on the time horizon $t \in [0,T]$, and that the low-fidelity model is the same system on $t \in [0,\tau]$ for $\tau<T$~\cite{Prangle2016}.
The low-fidelity and high-fidelity models can clearly be simulated independently on $[0,\tau]$ and $[0,T]$, respectively, and the low-fidelity simulation will, on average, be less computationally expensive than the high-fidelity simulation.
However, the natural coupling of the models fixes $y(t) = \tilde y(t)$ over $t \in [0, \tau]$, allowing the high-fidelity model to be simulated conditional on a low-fidelity simulation.
Many other possible couplings exist for different multifidelity models, often involving shared random noise processes, and methods for coupling are currently an area of active research~\cite{Croci2018,Lester2018,Prescott2020}.

In order to apply the multifidelity framework to ABC parameter inference, recall that each weight, $w_n$, generated by \Cref{ABC:Importance} requires a simulation $y_n \sim f(\cdot~|~\theta_n)$ from the high-fidelity model.
The multifidelity approach in \cite{Prescott2020} calculates the weight $w_n = w(\theta_n, \tilde y_n, u_n, y_n)$ by replacing the weighting function in \Cref{eq:ImportanceWeight} with
\begin{equation}
w(\theta, \tilde y, u, y)
= 
\frac{\pi(\theta)}{q(\theta)} \left( \mathbb I(\tilde y \in \Omega_{\epsilon}) + \frac{\mathbb I(u < \alpha(\theta, \tilde y))}{\alpha(\theta, \tilde y)} \left[ \mathbb I(y \in \Omega_{\epsilon}) - \mathbb I(\tilde y \in \Omega_{\epsilon}) \right] \right),
 \label{eq:w_mf}
\end{equation}
where $(\tilde y, y) \sim \check f(\cdot, \cdot~|~\theta)$ are coupled multifidelity simulations, where $u$ is a unit uniform random variable, and where $\alpha(\theta, \tilde y) \in (0,1]$ is a positive \emph{continuation probability}.
Note that, in general, we could allow the continuation probability, $\alpha(\theta, \tilde y, u, y)$, to depend on all of the stochastic variables, but we will assume $\alpha(\theta, \tilde y) \in (0,1]$ to be independent of $u$ and $y$.

The important consequence of the multifidelity weight in \Cref{eq:w_mf} is that, by the specific order in which we simulate the variables, the weight $w(\theta, \tilde y, u, y)$ may be calculated without the computational cost of simulating $y$.
Given $\theta$, we first simulate $\tilde y \sim \tilde f(\cdot~|~\theta)$ from the low-fidelity model.
This defines the continuation probability $\alpha(\theta, \tilde y) \in (0,1]$.
Second, we generate the unit uniform random variable, $u$.
If $u \geq \alpha(\theta, \tilde y) \in (0,1]$ then we can return $w(\theta, \tilde y, u, y)$ without simulating $y \sim f(\cdot ~|~\theta, \tilde y)$ from the coupling, thus incurring the lower computational expense of only simulating from $\tilde f$.
\Cref{MFABC:Importance} is an adaptation of \Cref{ABC:Importance} that returns a weighted sample $\{ \theta_n, w_n \}_{n=1}^N$ from the ABC posterior $p_\epsilon(\obs y ~|~ \theta)$, with weights calculated using $w$ in \Cref{eq:w_mf}.
As with \Cref{ABC:Importance}, in the rejection sampling case where $\hat q(\theta) = \pi(\theta)$, then we refer to this algorithm as MF-ABC-RS.

\begin{algorithm}
\caption{Multifidelity ABC importance sampling (MF-ABC-IS)}
\label{MFABC:Importance}
\begin{algorithmic}[1]
\REQUIRE{
Data $\obs y$ and neighbourhood $\Omega_\epsilon$; 
prior $\pi$; 
coupling $\check f(\cdot, \cdot~|~\theta)$ of models $f(\cdot~|~\theta)$ and $\tilde f(\cdot~|~\theta)$;
continuation probability function $\alpha = \alpha(\theta, \tilde y)$; 
sample index $n=0$; 
importance distribution $\hat q$ proportional to $q(\theta)$; 
stopping condition $S$.}
\ENSURE{Weighted sample $\{ \theta_n, w_n \}_{n=1}^{N}$.}
\REPEAT{}
 \STATE{Increment $n \leftarrow n+1$.}
 \STATE{Generate $\theta_n \sim \hat q(\cdot)$.}
 \STATE{Simulate $\tilde y_n \sim \tilde f(\cdot~|~\theta_n)$.}
 \STATE{Set $w_n = \mathbb I \left( \tilde y_n \in \Omega_{\epsilon} \right)$.}
 \STATE{Generate $u_n \sim \mathrm{Uniform}(0,1)$.}
 \IF{$u_n < \alpha(\theta_n, \tilde y_n)$}
  \STATE{Simulate $y_n \sim f(\cdot ~|~ \tilde y_n, \theta_n)$.}
  \STATE{Update $w_n \leftarrow w_n + \left[ \mathbb I(y_n \in \Omega_\epsilon) - w_n \right] \big/ \alpha(\theta_n, \tilde y_n)$.}
 \ENDIF{}
 \STATE{Update $w_n \leftarrow \left[ \pi(\theta_n) \big/ q(\theta_n) \right] w_n$.}
\UNTIL{$S = \texttt{true}$.}
\end{algorithmic}
\end{algorithm}


\begin{proposition}
\label{MFABCValidity}
The weighted Monte Carlo sample $\{ \theta_n, w_n\}$ returned by \Cref{MFABC:Importance} is from the ABC posterior $p_\epsilon(\theta~|~\obs y)$.
\end{proposition}
\begin{proof}
We note that the density of each $z = (\theta, \tilde y, u, y)$ sampled by \Cref{MFABC:Importance} is
$$
g(z) = \check f(\tilde y, y~|~\theta) \hat q(\theta),
$$
on $\mathcal Z = \Theta \times \mathcal Y \times [0,1] \times \mathcal Y$.
Furthermore, the multifidelity weight in \Cref{eq:w_mf} integrates such that
$$
\int_0^1 w(z) ~\mathrm du = \frac{\pi(\theta)}{q(\theta)} \mathbb I(y \in \Omega_\epsilon),
$$
which is independent of $\tilde y$. 
Therefore, for any integrable $F:\Theta \rightarrow \mathbb R$, we have the identity
\begin{align*}
\int_{\Theta} F(\theta) \pi_\epsilon(\theta~|~\obs y) ~\mathrm d\theta
    &= 
    \frac{Z_q}{Z} \int_{\Theta \times \mathcal Y} F(\theta) \frac{\pi(\theta)}{q(\theta)} \mathbb I(y \in \Omega_\epsilon) f(y~|~\theta) \hat q(\theta) ~\mathrm d\theta ~\mathrm dy \\
    &=
    \frac{Z_q}{Z} \int_{\Theta \times \mathcal Y^2} F(\theta) \frac{\pi(\theta)}{q(\theta)} \mathbb I(y \in \Omega_\epsilon) \check f(\tilde y, y~|~\theta) \hat q(\theta) ~\mathrm d\theta ~\mathrm dy ~\mathrm d\tilde y \\
    &=
    \frac{Z_q}{Z} \int_{\mathcal Z} F(\theta) w(z) g(z) ~\mathrm dz.
\end{align*}
Thus the normalised Monte Carlo estimate in \Cref{eq:estimate} is consistent for a weighted sample from \Cref{MFABC:Importance}.
\end{proof}

The key issue for \Cref{MFABC:Importance} (MF-ABC-IS) is the choice of continuation probability, $\alpha(\theta, \tilde y) \in (0,1]$.
Smaller values for $\alpha$ lead to greater computational savings, since high-fidelity simulations are generated less often.
However, it is possible to show that this comes at the cost of an increase in the mean squared error (MSE) of the estimators, $\bar F$, of integrable functions $F:\Theta \rightarrow R$.
In \Cref{s:Performance}, we quantify this tradeoff through the definition of the algorithm's efficiency.
\begin{note}
The multifidelity weight in \Cref{eq:w_mf} compares a uniform $u \in [0,1]$ with $\alpha(\theta, \tilde y)$.
In more generality, we may consider any multifidelity weight $w(\theta, \tilde y, u, y)$ that depends on $u \in \mathcal U$ with any distribution, $p(u~|~\theta, \tilde y)$, that does not rely on the expensive high-fidelity model.
If $w$ is designed such that
\[ \int_{u \in U} w(\theta, \tilde y, u, y) p(u~|~\theta, \tilde y) ~\mathrm du = \frac{\pi(\theta)}{q(\theta)} \mathbb I(y \in \Omega_\epsilon), \]
then, by an extension of \Cref{MFABCValidity}, the multifidelity weight remains valid.
\end{note}

\subsection{ABC Performance}
\label{s:Performance}

In \Cref{s:Background} so far, we have summarised previous developments of ABC algorithms that aim to improve performance, but have not yet defined how to quantify this improvement.
For any weighted sample $\{\theta_n, w_n\}$ built using \Cref{ABC:Importance} (ABC-IS) or \Cref{MFABC:Importance} (MF-ABC-IS), each sampled pair $(\theta_n, w_n)$ also incurs a computational cost, which we will denote $T_n$, producing a total computational cost of $T_{\mathrm{total}} = \sum_n T_n$ for the sample.
We define the observed efficiency of such a sample as follows.

\begin{definition}
\label{def:efficiency:obs}
The \emph{effective sample size}~\cite{Elvira2018,Liu2004} of a weighted Monte Carlo sample $\{ \theta_n, w_n \}$ is 
\[
\mathrm{ESS} = \frac{\left(\sum_n w_n\right)^2}{\sum_n w_n^2}.
\]
Given a computational cost of $T_n$ for each sampled pair $(w_n, \theta_n)$, the \emph{observed efficiency} of a weighted Monte Carlo sample is
\[
\frac{\mathrm{ESS}}{T_{\mathrm{total}}} = \frac{\left(\sum_n w_n\right)^2}{\left( \sum_n w_n^2 \right) \left( \sum_n T_n \right)},
\]
which is expressed in units of effective samples per time unit.
\end{definition}

Larger values of ESS typically correspond to smaller values of MSE in estimates of the form in \Cref{eq:estimate}.
Since ESS and $T_{\text{total}}$ both scale linearly with $N$, taking limits as $N \rightarrow \infty$ in the observed efficiency motivates the following definition of the theoretical efficiency of an ABC algorithm.

\begin{definition}
\label{def:efficiency:th}
The \emph{theoretical efficiency} of an ABC algorithm generating a weighted Monte Carlo sample $\{ \theta_n, w_n \}$ with simulation times $\{ T_n \}$ is
\[
	\psi = \frac{\mathbb E(w)^2}{\mathbb E(w^2) \mathbb E(T)},
\]
which is expressed in units of effective samples per time unit.
\end{definition}

Note that the theoretical efficiency is not just a characteristic of the ABC algorithm and of the models, but also of the numerical implementation and hardware of the computers generating the simulation.
For example, the theoretical efficiency of \Cref{ABC:Importance} (ABC-IS) is
\[
\psi_{\text{ABC-IS}} = \frac{Z}{\mathbb E_{\pi_\epsilon}(\pi / \hat q) \mathbb E_{\hat q}(T)},
\]
where $Z = p(y \in \Omega_\epsilon)$ is the normalisation constant in \Cref{eq:ABCposterior}, and where $\mathbb E_\nu$ denotes expectations with respect to probability density $\nu$.
Setting $\hat q = \pi$, the theoretical efficiency of ABC-RS is
\[
\psi_{\text{ABC-RS}} = \frac{Z}{\mathbb E_{\pi}(T)}.
\]
Hence, the theoretical efficiency of ABC-IS is improved over ABC-RS by choosing an importance distribution, $\hat q$, that is more likely than $\pi$ to propose $\theta$ incurring smaller simulation times and high posterior likelihoods.
In the case of \Cref{MFABC:Importance} with $\hat q(\theta) = \pi(\theta)$ (MF-ABC-RS), this performance measure has been used to determine a good choice of continuation probability, $\alpha(\theta, \tilde y)$~\cite{Prescott2020}.

By using the theoretical efficiency, $\psi$, as a performance metric in the remainder of this paper, we will quantify the improvement in performance over \Cref{ABC:Importance,ABC:SMC,MFABC:Importance} that can be achieved by combining multifidelity and SMC techniques.

\section{Multifidelity ABC-SMC}
\label{s:MF-ABC-SMC}
There are two distinct approaches to improving the performance of ABC parameter inference specified in \Cref{s:Background}.
ABC-SMC proposes parameters $\theta_n$ from a sequence of importance distributions that progressively approximate the target ABC approximation to the posterior.
In contrast, MF-ABC enables sampled parameters to be weighted without necessarily having to produce a simulation, $y_n$, from the high-fidelity model.
In this section we present the main contribution of this paper, which is to combine these orthogonal approaches into an MF-ABC-SMC algorithm.

We will replicate the procedure of extending \Cref{ABC:Importance} (ABC-IS) into \Cref{ABC:SMC} (ABC-SMC) in the multifidelity context, focusing first on the $O(N^2)$ SMC sampler based on sequential importance sampling.
Suppose that, as in \Cref{s:SMC}, we have a decreasing sequence of thresholds, $\epsilon_1 > \dots > \epsilon_T = \epsilon$, inducing the neighbourhoods, $\Omega_{\epsilon_1} \supseteq \cdots \supseteq \Omega_{\epsilon_T}$, and a sequence of ABC posteriors, $p_{\epsilon_t}(\theta~|~\obs y)$.
In principle, we can replace the call of \Cref{ABC:Importance} (ABC-IS) in step 2 of \Cref{ABC:SMC} with a call of \Cref{MFABC:Importance} (MF-ABC-IS) instead, and the SMC algorithm would proceed in much the same way.
However, the key difficulty with implementing multifidelity ABC-SMC lies in the definition of the importance distribution, $\hat q_{t+1}$, for generation $t+1$, given the Monte Carlo sample, $\{ \theta_n^{(t)}, w_n^{(t)} \}_{n=1}^{N_t}$, returned in generation $t$. 

Since the weights, $w_n$, are now calculated using the multifidelity weight in \Cref{eq:w_mf}, there is a positive probability that there exist $w_n < 0$.
Negative values of $w_n$ are generated whenever $\tilde y_n \in \Omega_\epsilon$, but where we also simulate the high-fidelity model such that $y_n \notin \Omega_\epsilon$.
This leads to two problems with the existing definition of the importance distribution in \Cref{eq:importance}.
Primarily, there may exist $\theta$ in the prior support with $q_{t+1}(\theta) < 0$, unless the perturbation kernels, $K_t$, are carefully designed to avoid this.
Secondly, even if the $K_t$ could be chosen to guarantee $q_{t+1}(\theta)>0$ on the prior support, it is not clear that we can easily sample from the importance distribution $\hat q_{t+1} \propto q_{t+1}$ when some $\theta_n^{(t)}$ have negative weights.

\begin{algorithm}
\caption{Multifidelity ABC-SMC with pre-determined $\alpha_t$ (MF-ABC-SMC-$\alpha$)}
\label{MFABC:SMCalpha}
\begin{algorithmic}[1]
\REQUIRE{
Data $\obs y$;
sequence of nested neighbourhoods $\Omega_{\epsilon_T} \subseteq \Omega_{\epsilon_{T-1}} \subseteq \cdots \subseteq \Omega_{\epsilon_1}$ for $0 < \epsilon = \epsilon_T < \epsilon_{T-1} < \dots < \epsilon_1$; 
prior $\pi$; 
coupling $\check f(\cdot, \cdot~|~\theta)$ of models $f(\cdot~|~\theta)$ and $\tilde f(\cdot~|~\theta)$;
initial importance distribution $\hat r_1$ (often set to $\pi$); 
perturbation kernels $K_t(\cdot~|~\theta)$; 
continuation probabilities $\alpha_t(\theta, \tilde y)$;
stopping conditions $S_t$;
where $t=1,\dots,T$.
}
\ENSURE{
Weighted sample $\{\theta_n^{(T)}, w_n^{(T)} \}_{n=1}^{N_T}$.
}
\FOR{$t = 1, \dots, T-1$}
 \STATE{Produce $\{ \theta_n^{(t)}, w_n^{(t)}  \}_{n=1}^{N_t}$ from \Cref{MFABC:Importance} (MF-ABC-IS), using the neighbourhood $\Omega_{\epsilon_t}$, continuation probability $\alpha_t$, importance distribution $\hat r_t$, and stopping condition $S_t$.
 }
 \STATE{Set $\hat w_n^{(t)} = | w_n^{(t)} |$ for $n=1,\dots,N_t$.}
 \STATE{Define $\hat r_{t+1}$ proportional to $r_{t+1}$ given in \Cref{eq:new_importance} }
\ENDFOR{}
\STATE{Produce $\{ \theta_n^{(T)}, w_n^{(T)} \}_{n=1}^{N_T}$ from \Cref{MFABC:Importance}, using neighbourhood $\Omega_{\epsilon}$, continuation probability $\alpha_T$, importance distribution $\hat r_T$ and stopping condition $S_T$.}

\end{algorithmic}
\end{algorithm}

In \Cref{MFABC:SMCalpha} (MF-ABC-SMC-$\alpha$) we have adapted the SMC approach to counter the possibility of negative weights, by considering a sampling algorithm that produces non-negative weights $\hat w_n^{(t)} \geq 0$ in parallel with $w_n^{(t)}$.
This method closely parallels the approach to the sign problem taken in \cite{Lyne2015}.
At each generation we produce two weighted samples at once, $\{ \hat w_n^{(t)}, \theta_n^{(t)} \}$ and $\{ w_n^{(t)}, \theta_n^{(t)} \}$.
We replace the importance distribution $\hat q_{t+1}$ defined by $q_{t+1}$ in \Cref{eq:importance} with the importance distribution $\hat r_{t+1}$ proportional to the non-negative function
\begin{equation}
r_{t+1}(\theta) = \begin{cases}
    \sum_{n=1}^N \hat w_n^{(t)} K_t(\theta~|~\theta_n^{(t)}) \bigg/ \sum_{m=1}^N \hat w_m^{(t)} & \pi(\theta)>0,
    \\
    0 & \text{else,}
    \end{cases}
\label{eq:new_importance}
\end{equation}
which is defined by the weights $\hat w_n^{(t)} = | w_n^{(t)} |$. 
With this choice of importance distribution, the weighted samples $\{ \hat w_n^{(t)}, \theta_n^{(t)} \}$ can be shown to be drawn from the alternative target distribution,
\begin{equation}
\label{eq:new_target}
\rho_t(\theta) \propto p_{\epsilon_t}(\theta~|~\obs y) + \delta_t(\theta),
\end{equation}
where the difference between the new target distribution, $\rho_t$, and the ABC posterior, $p_{\epsilon_t}$, is given by
\[
\delta_t(\theta) =  2\pi(\theta) \int_{\mathcal Y^2} (1-\alpha_t(\theta, \tilde y)) \mathbb I(\tilde y \in \Omega_{\epsilon_t}) \mathbb I(y \notin \Omega_{\epsilon_t}) \check f(\tilde y, y~|~\theta) ~\mathrm d\tilde y ~\mathrm dy,
\]
for $t=1,\dots,T$. 
The function $\delta_t$ implies that the new target distribution (compared to $\pi_{\epsilon_t}$) contains additional density in regions of parameter space where it is more likely that $\tilde y \in \Omega_{\epsilon_t}$ but $y \notin \Omega_{\epsilon_t}$; in other words, where $w_n^{(t)} < 0$ is more likely.

The importance distribution in \Cref{eq:new_importance} effectively makes the SMC algorithm target $\rho_t$ at each generation instead of $p_{\epsilon_t}$.
However, the weighted samples $\{ w_n^{(t)}, \theta_n^{(t)} \}$ based on the multifidelity weights $w_n^{(t)}$ from each generation's weighting function,
\begin{equation}
w_t(\theta, \tilde y, u, y)
= 
\frac{\pi(\theta)}{r_t(\theta)} \left( \mathbb I(\tilde y \in \Omega_{\epsilon_t}) + \frac{\mathbb I(u < \alpha_t(\theta, \tilde y))}{\alpha_t(\theta, \tilde y)} \left[ \mathbb I(y \in \Omega_{\epsilon_t}) - \mathbb I(\tilde y \in \Omega_{\epsilon_t}) \right] \right),
\label{eq:w_mf_t}
\end{equation}
remain from the ABC posteriors $p_{\epsilon_t}$.
Hence, at any generation (and in particular at $t=T$), we can produce an estimate of a $p_{\epsilon_t}$-integrable function $F:\Theta \rightarrow \mathbb R$, such that
\[
\mathbb E_{p_{\epsilon_t}} (F) \approx \frac{\sum_n w_n^{(t)} F(\theta_n^{(t)})}{ \sum_m w_m^{(t)}},
\]
is a consistent Monte Carlo estimate of $F$ under the ABC posterior.

\begin{note}
\label{Note:ON2}
Each calculation of the weight in \Cref{eq:w_mf_t} relies on the $O(N)$ calculation of the importance weight in \Cref{eq:new_importance}, making the SMC sampler $O(N^2)$.
An alternative sampling method that is linear in $N$ is proposed in \cite{DelMoral2012}, which replaces the sequential importance sampling approach described above.
However, as noted in \cite{DelMoral2012}, both the $O(N)$ and $O(N^2)$ sampling algorithms require $O(N)$ simulations of $(\tilde y, y) \sim \check f(\cdot, \cdot~|~\theta)$ or $\tilde y \sim \tilde f(\cdot~|~\theta)$.
In this multifidelity setting, we are assuming that simulation time comprises the vast majority of the computational burden of each calculation of \Cref{eq:w_mf_t}.
The benefit of the multifidelity weight is that it reduces the computational burden of generating $(\tilde y, y)$ by sometimes requiring $\tilde y$ alone.
In \Cref{LinearSMCSampler} we show that the $O(N)$ SMC sampling algorithm dilutes this benefit.
Therefore, in this work we will focus on the sequential importance sampling SMC sampler described in \Cref{MFABC:SMCalpha}.
\end{note}

\section{Adaptive MF-ABC-SMC}
\label{MFABCSMC}

In \Cref{MFABC:SMCalpha} (MF-ABC-SMC-$\alpha$), in addition to an assumed sequence of ABC thresholds, $\epsilon_t$, perturbation kernels, $K_t$, and stopping conditions $S_t$, for $t=1,\dots,T$, we also assume a given sequence of continuation probabilities, $\alpha_t$.
In this algorithm, each importance distribution, $\hat r_{t+1}$, is determined by the output at generation $t$.
Methodologies for adaptively choosing the perturbation kernels, $K_{t+1}$, and the ABC thresholds, $\epsilon_{t+1}$, based on the preceding generations' samples, have been explored in previous work~\cite{DelMoral2012,Filippi2013}.
In this section, we will consider the adaptive approach to choosing each generation's continuation probability, $\alpha_{t+1}$, based on the simulation output of generation $t$.

\subsection{Optimal continuation probabilities}
\label{s:eta}

For each generation, $t$, the continuation probability, $\alpha_t$, is an input into that generation's call of \Cref{MFABC:Importance}. 
Dropping the generational indexing $t$ temporarily, in this subsection we first consider how to choose a continuation probability function, $\alpha(\theta, \tilde y)$, to maximise the theoretical efficiency, $\psi$, of any run of \Cref{MFABC:Importance} (MF-ABC-IS), as specified in \Cref{def:efficiency:th}.

For simplicity, we will constrain the search for optimal $\alpha(\theta, \tilde y)$ to the piecewise constant function
\begin{equation}
\label{eq:constantrates}
 \alpha(\theta, \tilde y) = \eta_1 \mathbb I(\tilde y \in \Omega_{\epsilon}) + \eta_2 \mathbb I(\tilde y \notin \Omega_{\epsilon}),
\end{equation}
for the constants $\eta_1, \eta_2 \in (0,1]$.
Here, $\eta_1$ is the probability of generating $y$ after a `positive' low-fidelity simulation (where $\tilde y \in \Omega_\epsilon$) and $\eta_2$ is the probability of generating $y$ after a `negative' low-fidelity simulation (where $\tilde y \notin \Omega_\epsilon$).
The goal of this section is to specify the values of the two parameters, $\eta_1$ and $\eta_2$, that will give the largest theoretical efficiency, $\psi$.
In previous work, we have derived the optimal values of $\eta_1$ and $\eta_2$ to use in the special case of \Cref{MFABC:Importance} (MF-ABC-RS) corresponding to rejection sampling, where the importance distribution $\hat q = \pi$ equal to the prior distribution~\cite{Prescott2020}.
We can now extend this analysis by finding optimal values of $\eta_1$ and $\eta_2$ to use in the more general case of \Cref{MFABC:Importance} (MF-ABC-IS).
The key to this optimisation is the following lemma, which describes how the efficiency of \Cref{MFABC:Importance} varies with the continuation probabilities used.
The lemma assumes a given importance distribution, $\hat q(\theta)$, defined as the normalisation of the known non-negative function $q(\theta)$.

\begin{lemma}
\label{lemma:phi}
The theoretical efficiency, given in \Cref{def:efficiency:th}, of \Cref{MFABC:Importance} (MF-ABC-IS) varies with the continuation probabilities $\eta_1$ and $\eta_2$ according to
\[
\psi(\eta_1,\eta_2) = \frac{\mathbb E(w)^2}{\mathbb E(w^2) \mathbb E(T)} = \frac{Z^2}{\phi(\eta_1, \eta_2)},
\]
where the denominator is expressed as a function of $(\eta_1, \eta_2)$ such that
\begin{equation}
\label{eq:Phi}
\phi(\eta_1, \eta_2) = 
\left(  W + \left( \frac{1}{\eta_1} - 1 \right) \fp W + \left( \frac{1}{\eta_2} - 1 \right) \fn W \right)
\left( \bar{T}_{\mathrm{lo}}
 + \eta_1 \bar T_{\mathrm{hi}, \mathrm p}
 + \eta_2 \bar T_{\mathrm{hi}, \mathrm n} \right).
\end{equation}
The coefficients in $\psi$ are given by the integrals
\begin{subequations}
\label{eq:PhiComponents}
\begin{align}
\label{eq:Z}
Z &= \int L_\epsilon(\theta) \pi(\theta) ~\mathrm d\theta
,&
&L_\epsilon(\theta) = \int_{\mathcal Y^2} \mathbb I(y \in \Omega_\epsilon) \check f(\tilde y, y~|~\theta) ~\mathrm d\tilde y ~\mathrm dy
,\\
W &= \int \frac{\pi(\theta)}{q(\theta)} L_\epsilon(\theta) \pi(\theta) ~\mathrm d\theta
,\\
\fp W &= \int \frac{\pi(\theta)^2}{q(\theta)} \fp p(\theta) ~\mathrm d\theta 
,&
&\fp p(\theta) = \int_{\mathcal Y^2} \mathbb I(\tilde y \in \Omega_\epsilon) \mathbb I(y \notin \Omega_\epsilon) \check f(\tilde y, y~|~\theta) ~\mathrm d\tilde y ~\mathrm dy 
,\\
\fn W &= \int \frac{\pi(\theta)^2}{q(\theta)} \fn p(\theta) ~\mathrm d\theta 
,&
&\fn p(\theta) = \int_{\mathcal Y^2} \mathbb I(\tilde y \notin \Omega_\epsilon) \mathbb I(y \in \Omega_\epsilon) \check f(\tilde y, y~|~\theta) ~\mathrm d\tilde y ~\mathrm dy 
,\\
\bar T_{\mathrm{lo}} &= \int T_{\mathrm{lo}}(\theta) q(\theta) ~\mathrm d\theta
,&
&T_{\mathrm{lo}}(\theta) = \int_{\mathcal Y^2} T(\tilde y) \check f(\tilde y, y~|~\theta) ~\mathrm d\tilde y ~\mathrm dy
,\\
\bar T_{\mathrm{hi,p}} &= \int T_{\mathrm{hi,p}}(\theta) q(\theta) ~\mathrm d\theta
,&
&T_{\mathrm{hi,p}}(\theta) = \int_{\mathcal Y^2} T(y) \mathbb I(\tilde y \in \Omega_\epsilon) \check f(\tilde y, y~|~\theta) ~\mathrm d\tilde y ~\mathrm dy
,\\
\bar T_{\mathrm{hi,n}} &= \int T_{\mathrm{hi,n}}(\theta) q(\theta) ~\mathrm d\theta
,&
&T_{\mathrm{hi,n}}(\theta) = \int_{\mathcal Y^2} T(y) \mathbb I(\tilde y \notin \Omega_\epsilon) \check f(\tilde y, y~|~\theta) ~\mathrm d\tilde y ~\mathrm dy
,
\end{align}
\end{subequations}
where $T(\tilde y)$ is the computational cost of simulating $\tilde y \sim \tilde f(\cdot~|~\theta)$ and where $T(y)$ is the cost of simulating $y \sim f(\cdot~|~\theta, \tilde y)$, such that $T(\tilde y)+T(y)$ is the cost of simulating $(\tilde y, y) \sim \check f(\cdot, \cdot~|~\theta)$.
\end{lemma}

\subsubsection{Optimising efficiency}
We can conclude from \Cref{lemma:phi} that the optimal continuation probabilities $(\eta_1^\star, \eta_2^\star)$ in any closed, bounded domain $\mathcal H \subseteq (0,1]^2$ are those that minimise the function $\phi(\eta_1, \eta_2)$ given in \Cref{eq:Phi}.
\Cref{etastar:unbounded,etastar:boundary} below explicitly find the global minimiser of $\phi$ over $[0,\infty)^2$ (if it exists), and then over the boundary $\partial \mathcal H$ of a rectangular domain $\mathcal H = [\rho_1, 1] \times [\rho_2, 1]$, where the user-specified lower bounds $\rho_1$ and $\rho_2$ are chosen to ensure $\mathcal H$ is closed.
These results combine in \Cref{etastar} to give the minimiser of $\phi$ over $\mathcal H$, and hence the optimal continuation probabilities for use in \Cref{MFABC:Importance}.

\begin{lemma}
\label{etastar:unbounded}
We first consider all non-negative values of $\eta_1, \eta_2 \geq 0$. 
If $W > \fp{W} + \fn{W}$, then the minimum value of $\phi(\eta_1,\eta_2)$ in \Cref{eq:Phi}, and the optimal value of $(\eta_1, \eta_2)$ in the entire positive quadrant, are given by
\begin{subequations}
\label{eq:etastar:unbounded}
\begin{align}
\bar \phi &= \left( \sqrt{(W - \fp{W} - \fn{W}) \bar T_{\mathrm{lo}}} + \sqrt{\fp{W} \bar T_{\mathrm{hi, p}}} + \sqrt{\fn{W} \bar T_{\mathrm{hi, n}}} \right)^2, \\
\left( \bar \eta_1, \bar \eta_2 \right) &= \left( 
    \sqrt{ \frac{\bar T_{\mathrm{lo}}}{W - \fp{W} - \fn{W}} \cdot \frac{\fp{W}}{\bar T_{\mathrm{hi, p}} }},
    \sqrt{ \frac{\bar T_{\mathrm{lo}}}{W - \fp{W} - \fn{W}} \cdot \frac{\fn{W}}{\bar T_{\mathrm{hi, n}} }}
\right),
\end{align}
\end{subequations}
respectively.
If $W \leq \fp{W} + \fn{W}$, then there is no minimum of $\phi(\eta_1,\eta_2)$ in $\eta_1, \eta_2 \geq 0$.
\end{lemma}

\begin{lemma}
\label{etastar:boundary} 
Under the same conditions as \Cref{etastar:unbounded}, fix the closed region $\mathcal H = [\rho_1, 1] \times [\rho_2, 1]$ of positive continuation probabilities with user-defined lower bounds $\rho_1, \rho_2 \in (0,1)$.
Define the two functions, for $x>0$,
\begin{subequations}
\label{eq:etastar:boundary}
\begin{align}
\eta_1(x) &= \max \left\{ \rho_1, ~\min \left[ 1, 
\sqrt{\frac{\bar T_{\mathrm{lo}} + \bar T_{\mathrm{hi}, \mathrm{n}} x }{W - \fp{W} - (1-x^{-1})\fn{W}} \cdot \frac{\fp{W}}{\bar T_{\mathrm{hi}, \mathrm p}}} ~
\right] \right\}, \\
\eta_2(x) &= \max \left\{ \rho_2, ~\min \left[ 1, 
\sqrt{\frac{\bar T_{\mathrm{lo}} + \bar T_{\mathrm{hi}, \mathrm{p}} x }{W - (1-x^{-1}) \fp{W} - \fn{W}} \cdot \frac{\fn{W}}{\bar T_{\mathrm{hi}, \mathrm n}}} ~
 \right] \right\}.
\end{align}
\end{subequations}
Then the minimum value of $\phi$ on the boundary, $\partial \mathcal H$, of $\mathcal H$ is attained at the minimum of $\phi(1, \eta_2(1))$, $\phi(\eta_1(1), 1)$, $\phi(\rho_1, \eta_2(\rho_1))$ or $\phi(\eta_1(\rho_2), \rho_2)$.
\end{lemma}
 
\begin{proposition}
\label{etastar}
 Assume the same conditions as \Cref{etastar:unbounded,etastar:boundary}. 
 Compute the minimiser, $(\bar \eta_1, \bar \eta_2)$, and minimal value, $\bar \phi$, of $\phi$ in $(0,\infty)^2$ using \Cref{etastar:unbounded}, if they exist.
 If $(\bar \eta_1, \bar \eta_2) \in \mathcal H$ then set $(\eta_1^\star, \eta_2^\star) = (\bar \eta_1, \bar \eta_2)$ and $\phi^\star = \bar \phi$. 
 Otherwise, set $\phi^\star$ equal to the minimum of the four values of $\phi$ listed in \Cref{etastar:boundary}, and $(\eta_1^\star, \eta_2^\star)$ to the associated argument.
Then $\phi^\star$ is the minimum value of $\phi$ over $(\eta_1, \eta_2) \in \mathcal H$, and $(\eta_1^\star, \eta_2^\star)$ are the minimising continuation probabilities.
\end{proposition}

\subsubsection{Interpreting efficiency}
The optimised efficiency of \Cref{MFABC:Importance} is determined by the values of the various coefficients defined in \Cref{eq:PhiComponents}.
The normalisation constant, $Z$, and the ABC approximation to the likelihood, $L_\epsilon$, are properties of the ABC approach as defined in \Cref{eq:ABC}.
The coefficient $W$ can be written $W = Z \mathbb E_{\pi}(p_\epsilon/q)$, as a scaling of the prior expectation of the ratio $p_\epsilon / q$.
Larger values of $W$ require $q$ to be more concentrated (relative to $\pi$) in regions of high posterior density, which is a well-known characteristic of importance sampling~\cite{Owen2013}.
Thus, neither $Z$ nor $W$ relate specifically to the multifidelity approach and are properties of ABC importance sampling.

The coefficient $\bar T_{\mathrm{lo}}$ represents the average time taken to simulate $\tilde y$ from the low-fidelity model.
The other two time-based coefficients, $\bar T_{\mathrm{hi,p}}$ and $\bar T_{\mathrm{hi,n}}$, represent the average time taken to complete the simulation of $(\tilde y, y) \sim \check f$ from the coupling, conditional on whether or not $\tilde y \in \Omega_\epsilon$.
These coefficients therefore determine how much time can be saved by avoiding expensive simulations and stopping after generating $\tilde y \sim \tilde f(\cdot~|~\theta)$.

The key determinants of the success of the multifidelity technique are $\fp W$ and $\fn W$ and their tradeoff between the high-fidelity simulation costs, $\bar T_{\mathrm{hi,p}}$ and $\bar T_{\mathrm{hi,n}}$.
\Cref{eq:Phi} implies that the marginal cost to the efficiency of decreasing $\eta_1$ and $\eta_2$ is smaller for smaller values of $\fp W$ and $\fn W$.
These coefficients can be written as the two prior expectations,
\[
 \fp W 
 = 
 \mathbb E_\pi \left( \frac{\pi}{q} ~\fp p \right),
 \quad
 \fn W 
 = 
 \mathbb E_\pi \left( \frac{\pi}{q} ~\fn p \right).
\]
Thus, they are scalings of the probabilities of a false positive (where $\tilde y \in \Omega_\epsilon$ but $y \notin \Omega_\epsilon$) and a false negative (where $\tilde y \notin \Omega_\epsilon$ but $y \in \Omega_\epsilon$), respectively.
Hence, small values of $\fp W$ and $\fn W$ correspond to at least one of the following cases.
First, if the low-fidelity and high-fidelity models are closely correlated, then the probability of a false positive or false negative is small.
This demonstrates the value of coupling the models, as described in \Cref{s:MFABC}.
Second, for small values of $\fp W$ and $\fn W$ we require $q$ to be larger than $\pi$ in regions of parameter space where false positives or false negatives are relatively likely.
We can intepret this as a requirement that (in addition to $q$ being concentrated in regions of high posterior density) the region of parameter space where simulations of the low-fidelity and high-fidelity model are less often in agreement (in terms of membership of $\Omega_\epsilon$) should be explored more thoroughly by $q$ than by $\pi$.

\subsection{Constructing continuation probabilities}
\label{s:eta:MC}
There is an important barrier to implementing \Cref{etastar} as a method for choosing optimal continuation probabilities.
Before running any ABC iterations, and in the absence of extensive analysis of the models being simulated, the quantities in \Cref{eq:PhiComponents} are unknown.
We therefore cannot directly construct the optimisers in \Cref{eq:etastar:unbounded,eq:etastar:boundary}.
However, recall that we are aiming to use this method in the context of sequential Monte Carlo to adaptively produce continuation probabilities.
In \Cref{MFABC:SMCalpha}, at each generation $t \geq 1$, we have a sample $\{ \theta_n^{(t)}, w_n^{(t)} \}$ that is used to produce $\hat r_{t+1}$.
The proposed approach to constructing continuation probabilities is similar: we will use the same Monte Carlo sample to also produce approximately optimal values of $\eta_1$ and $\eta_2$ defining the continuation probability $\alpha_{t+1}$ to use at generation $t+1$.
The following definition specifies how to calculate approximations of the quantities in \Cref{eq:PhiComponents} using an existing Monte Carlo sample.
These approximations can then be substituted into \Cref{eq:etastar:unbounded,eq:etastar:boundary}.
Hence, we can estimate the optimal continuation probabilities as given by \Cref{etastar}.

\begin{definition}
\label{def:MonteCarlo}
Consider a Monte Carlo sample $\left\{ \theta_n, w_n \right\}_{n=1}^{N}$ constructed from a run of \Cref{MFABC:Importance}, which used the importance distribution $\hat q(\theta)$ proportional to $q(\theta)$ and continuation probability $\alpha(\theta, \tilde y)$.
For the set of low-fidelity simulations $\{ \tilde y_n \}_{n=1}^{N}$ and high-fidelity simulations $\{ y_n \}_{n \in M}$, where $M = \{ n~:~ y_n \text{ exists} \}$, store: the simulation times, $\tilde t_n = T(\tilde y_n)$ and $t_n = T(y_n)$; the distances from data, $\tilde d_n = d(\tilde y_n, \obs y)$ and $d_n = d(y_n, \obs y)$; the importance densities $q_n = q(\theta_n)$; and the continuation probabilities $\alpha_n = \alpha(\theta_n, \tilde y_n)$.

We now consider finding the optimal continuation probability, $\alpha^\star$, to be used in a new run of \Cref{MFABC:Importance}, with the new importance distribution, $q^\star(\theta)$, and the new ABC threshold, $\epsilon$.
We define the Monte Carlo estimates,
\begin{subequations}
\label{eq:MonteCarlo}
\begin{align}
 \hat Z &= \frac{1}{N} \left[ \sum_{n=1}^N \frac{\pi(\theta_n)}{q_n} \mathbb I(\tilde d_n < \epsilon) + \sum_{n \in M} \frac{\pi(\theta_n)}{q_n \alpha_n} \left( \mathbb I(d_n < \epsilon) - \mathbb I(\tilde d_n < \epsilon)\right) \right]
 \label{eq:Zhat}
, \\
 \hat{W} &= \frac{1}{N} \left[ \sum_{n=1}^N \frac{\pi(\theta_n)^2}{q^\star(\theta_n) q_n} \mathbb I(\tilde d_n < \epsilon) + \sum_{n \in M} \frac{\pi(\theta_n)^2}{q^\star(\theta_n) q_n \alpha_n} \left( \mathbb I(d_n < \epsilon) - \mathbb I(\tilde d_n < \epsilon)\right) \right]
 \label{eq:MC_tpfn}
, \\ \fp{\hat{W}} &= \frac{1}{N} \sum_{n \in M} \frac{ \pi(\theta_n)^2}{q^\star(\theta_n) q_n \alpha_n} \mathbb I (\tilde d_n < \epsilon) \mathbb I (d_n \geq \epsilon)
 \label{eq:MC_fp}
, \\ \fn{\hat{W}} &= \frac{1}{N} \sum_{n \in M} \frac{ \pi(\theta_n)^2}{q^\star(\theta_n) q_n \alpha _n} \mathbb I (\tilde d_n \geq \epsilon) \mathbb I (d_n < \epsilon)
 \label{eq:MC_fn}
, \\ \hat T_{\mathrm{lo}} &= \frac{1}{N} \sum_{n=1}^{N} \frac{q^\star(\theta_n)}{q_n} \tilde t_n
 \label{eq:T_lo}
, \\ \hat T_{\mathrm{hi}, \mathrm p} &= \frac{1}{N} \sum_{n \in M} \frac{q^\star(\theta_n)}{q_n \alpha_n} \mathbb I(\tilde d_n < \epsilon) t_n
 \label{eq:T_hi_p}
, \\ \hat T_{\mathrm{hi}, \mathrm n} &= \frac{1}{N} \sum_{n \in M} \frac{q^\star(\theta_n)}{q_n \alpha_n} \mathbb I(\tilde d_n \geq \epsilon) t_n
 \label{eq:T_hi_n}
,
\end{align}
\end{subequations}
corresponding to the quantities in \Cref{eq:PhiComponents}.
\end{definition}

The estimates in \Cref{eq:MonteCarlo} are scaled Monte Carlo estimates of the quantities in \Cref{eq:PhiComponents}, such that the approximation
\[
\psi(\eta_1,\eta_2) = \frac{Z^2}{\phi(\eta_1,\eta_2)} \approx \frac{\hat Z^2}{\hat \phi(\eta_1,\eta_2)},
\]
holds, with
\begin{equation}
\label{eq:PhiHat}
\hat \phi(\eta_1, \eta_2) = 
\left(  \hat W + \left( \frac{1}{\eta_1} - 1 \right) \fp{\hat W} + \left( \frac{1}{\eta_2} - 1 \right) \fn{\hat W} \right)
\left( \hat{T}_{\mathrm{lo}}
 + \eta_1 \hat T_{\mathrm{hi}, \mathrm p}
 + \eta_2 \hat T_{\mathrm{hi}, \mathrm n} \right).
\end{equation}
Thus, we can substitute the estimates in \Cref{eq:MonteCarlo} into \Cref{eq:etastar:unbounded,eq:etastar:boundary}.
Applying \Cref{etastar} with these estimates thus provides near-optimal continuation probabilities for the new run of \Cref{MFABC:Importance}, constructed from the existing Monte Carlo sample.

\begin{note}
The Monte Carlo estimates in \Cref{eq:MonteCarlo} are not independent of each other, so there is a bias in the approximation $\hat Z^2 / \hat \phi$.
Hence, the continuation probabilities $(\eta_1^\star, \eta_2^\star)$ returned by \Cref{etastar}, if using the estimates in \Cref{eq:MonteCarlo}, can only be near-optimal. 
\end{note}

\subsection{Adaptive MF-ABC-SMC algorithm}
\label{s:MFABC:SMC}

\Cref{MFABC:SMCalpha} presents a multifidelity ABC-SMC algorithm that relies on a predetermined sequence of continuation probability functions, $\alpha_t(\theta, \tilde y)$.
In \Cref{s:eta,s:eta:MC} we have shown how to use an existing Monte Carlo sample to produce continuation probabilities of the form
\[
\alpha_t(\theta, \tilde y) = \eta_1 \mathbb I(\tilde y \in \Omega_{\epsilon_t}) + \eta_2 \mathbb I(\tilde y \notin \Omega_{\epsilon_t}),
\]
where the values of $\eta_1$ and $\eta_2$ are chosen according to \Cref{etastar}, in order to (approximately) optimise the efficiency of generating a sample from \Cref{MFABC:Importance}.
This result allows us to write an adaptive MF-ABC-SMC algorithm that uses the Monte Carlo output of \Cref{MFABC:Importance} at generation $t$ to construct not just an importance distribution, $\hat r_{t+1}(\theta)$, but also a continuation probability, $\alpha_{t+1}(\theta, \tilde y)$, for use in the next generation.

\Cref{MFABC:SMC} (MF-ABC-SMC) is an adaptive multifidelity sequential Monte Carlo algorithm for ABC parameter inference.
In place of the pre-defined continuation probabilities $\alpha_t$ used in \Cref{MFABC:SMCalpha} (MF-ABC-SMC-$\alpha$), we instead only require an initial continuation probability, most sensibly set to $\alpha_1 \equiv 1$, and the two lower bounds, $\rho_1$ and $\rho_2$, on the allowed values of the continuation probabilities.
Since the final sample is generated by a run of \Cref{MFABC:Importance}, by \Cref{MFABCValidity} it follows that the weighted sample is from the ABC posterior, $p_\epsilon(\theta~|~\obs y)$.

\begin{algorithm}
\caption{Multifidelity ABC-SMC (MF-ABC-SMC)}
\label{MFABC:SMC}
\begin{algorithmic}[1]
\REQUIRE{
Data $\obs y$;
sequence of nested neighbourhoods $\Omega_{\epsilon_T} \subseteq \Omega_{\epsilon_{T-1}} \subseteq \cdots \subseteq \Omega_{\epsilon_1}$ for $0 < \epsilon = \epsilon_T < \epsilon_{T-1} < \dots < \epsilon_1$; 
prior $\pi$; 
coupling $\check f(\cdot, \cdot~|~\theta)$ of models $f(\cdot~|~\theta)$ and $\tilde f(\cdot~|~\theta)$;
initial importance distribution $\hat r_1$ (often set to $\pi$); 
initial continuation probability $\alpha_1 \equiv 1$;
lower bounds on continuation probabilities, $\rho_1,\rho_2 \in (0,1)$; 
perturbation kernels $K_t(\cdot~|~\theta)$; 
stopping conditions $S_t$;
where $t=1,\dots,T$.
}
\ENSURE{
Weighted sample $\{\theta_n^{(T)}, w_n^{(T)} \}_{n=1}^{N_T}$.
}
\FOR{$t = 1, \dots, T-1$}
 \STATE{Produce $\{ \theta_n^{(t)}, w_n^{(t)}  \}_{n=1}^{N_t}$ from \Cref{MFABC:Importance} (MF-ABC-IS), using the neighbourhood $\Omega_{\epsilon_t}$, continuation probability $\alpha_t$, importance distribution $\hat r_t$, and stopping condition $S_t$.
 Store simulation times, distances, importance densities and continuation probabilities as specified in \Cref{def:MonteCarlo}.}
 \STATE{Set $\hat w_n^{(t)} = | w_n^{(t)} |$ for $n=1,\dots,N_t$.}
 \STATE{Define $\hat r_{t+1}$ proportional to $r_{t+1}$ given in \Cref{eq:new_importance} }
 \STATE{Update the estimates in \Cref{eq:MonteCarlo} with the values stored at step 2, using importance distribution $r_{t+1}(\theta)$ and ABC threshold $\epsilon_{t+1}$.}
 \STATE{Calculate $(\eta_1^\star, \eta_2^\star)$ using \Cref{etastar} with lower bounds $\rho_1, \rho_2$.}
 \STATE{Set $\alpha_{t+1}(\tilde y, \theta) = \eta_1^\star \mathbb I(\tilde y \in \Omega_{\epsilon}) + \eta_2^\star \mathbb I(\tilde y \notin \Omega_{\epsilon})$. }
\ENDFOR{}
\STATE{Produce $\{ \theta_n^{(T)}, w_n^{(T)} \}_{n=1}^{N_T}$ from \Cref{MFABC:Importance}, using neighbourhood $\Omega_{\epsilon_T}$, continuation probability $\alpha_T$, importance distribution $\hat r_T$ and stopping condition $S_T$.}

\end{algorithmic}
\end{algorithm}

\begin{note}
In common with the non-adaptive \Cref{MFABC:SMCalpha} (MF-ABC-SMC-$\alpha$), the importance weights $r_t(\theta_n^{(t)})$ are required to construct the weights $w_n^{(t)}$ in step 2, each of which is an $O(N)$ calculation.
However, each calculation of $r_{t+1}(\theta_n^{(t)})$ required in step 5 is also $O(N)$.
Thus, there is extra cost at each generation on the order $O(N^2)$.
However, in common with \Cref{Note:ON2}, we will assume that this cost is justified by our aim to reduce the large simulation burden that dominates the algorithm's run time.
\end{note}

\section{Example: Kuramoto Oscillator Network}
\label{s:Example}

To demonstrate the multifidelity and SMC approaches to parameter inference, we will infer the parameters of a Kuramoto oscillator model on a complete network, with stochastic heterogeneity in each node's intrinsic frequency.
In \Cref{s:existing} we consider the performance of the previously developed ABC algorithms introduced in \Cref{s:Background} (ABC-RS, ABC-SMC and MF-ABC-RS) and demonstrate the orthogonal ways in which the SMC and multifidelity techniques improve performance.
In \Cref{s:Results} we apply the adaptive algorithm, \Cref{MFABC:SMC} (MF-ABC-SMC) to demonstrate that the efficiency of parameter estimation is significantly improved by combining the multifidelity and SMC approaches.
The algorithms have been implemented in Julia~\cite{Julia} and the source code can be found at \url{github.com/tpprescott/mf-abc-smc}.

The Kuramoto oscillator model is defined on a complete network of $M$ nodes, where each node, $i$, has a dynamically evolving phase value, $\phi_i$, determined by the ordinary differential equation
\begin{equation}
 \label{eq:Kuramoto_hi}
 \dot \phi_i  = \omega_i + \frac{K}{M} \sum_{j=1}^{M} \sin \left( \phi_j - \phi_i \right),
\end{equation}
for $i=1,\dots,M$.
Each constant $\omega_i$, the intrinsic angular velocity, is an independent draw from a Cauchy distribution with median $\omega_0$ and dispersion parameter $\gamma$. 
In addition to these two parameters, we have an interconnection strength $K$.
Simulations of the ODE system are run over a fixed time interval $t \in [0,T]$, and we will assume fixed initial conditions $\phi_i(0)=0$ for all $i$.

The multifidelity approach makes use of a low-dimensional approximation of the coupled oscillator dynamics, as described in~\cite{Hannay2018,Ott2008,Ott2009}.
The approximation is based on tracking the Daido order parameters, which are a set of complex-valued representations of the high-dimensional vector $(\phi_i)_{i=1}^M$, defined as
\[
 Z_n(t) = \frac{1}{M} \sum_{j=1}^M \exp( i n \phi_j),
\]
for positive integers $n$ and the imaginary unit $i$.
A system of coupled ODEs can be generated for the set of $Z_n$.
Under the assumption that $Z_n(t) = Z_1(t)^n$, known as the Ott-Antonsen ansatz~\cite{Hannay2018,Ott2008}, the system can be reduced to a single ODE for $Z_1$, which is known as the Kuramoto parameter.
This complex-valued trajectory is usually represented by two real trajectories, corresponding to its magnitude $R(t) = \| Z_1(t)) \|$ and phase $\Phi(t) = \arg(Z_1(t))$.
The approximation of the $M$-dimensional ODE system in \Cref{eq:Kuramoto_hi} under the OA ansatz is thus given by the two-dimensional ODE system
\begin{subequations}
\label{eq:Kuramoto_lo}
\begin{align}
 \dot{\tilde R} &= \left( \frac{K}{2} - \gamma \right) \tilde R - \frac{K}{2} \tilde R^3, \\
 \dot{\tilde \Phi} &= \omega_0,
\end{align}
\end{subequations}
with initial conditions $(\tilde R(0), \tilde \Phi(0)) = (1, 0)$, which directly simulates the low-dimensional representation of the $M$-dimensional state vector.

The goal of this example is to infer the parameters $(K, \omega_0, \gamma)$ based on synthetic data, generated by simulating a system of $M=256$ oscillators with random angular velocities $\omega_i$ over $t \in (0, 30]$.
We record the trajectories $\obs R(t)$ and $\obs \Phi(t)$ of the magnitude and phase of the Kuramoto parameter.
The parameter values used to generate these data are $(K=2,~\omega_0 = \pi/3,~\gamma=0.1)$.
The likelihood of the observed data under the model in \Cref{eq:Kuramoto_hi} with stochastic parameters $\omega_i \sim \mathrm{Cauchy}(\omega_0, \gamma)$ is unavailable, and we must therefore resort to ABC inference, requiring repeated simulation.

\begin{figure}
    \centering
    \subfloat[Magnitude: $R(t)$]{\includegraphics[width=0.45\textwidth]{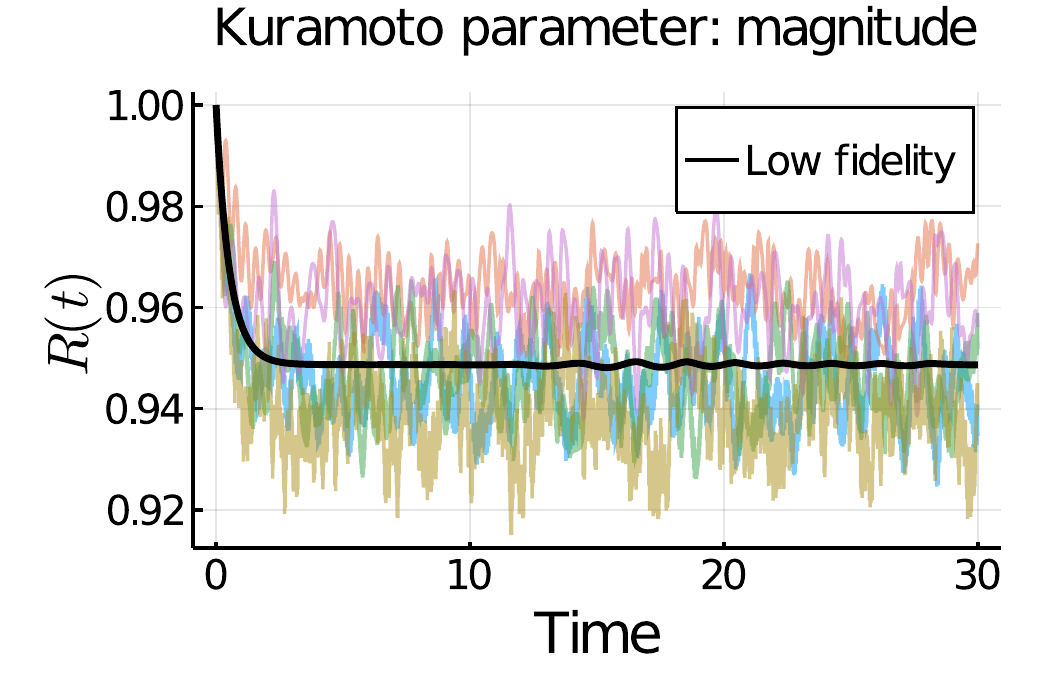}}~
    \subfloat[Phase: $\Phi(t)$]{\includegraphics[width=0.45\textwidth]{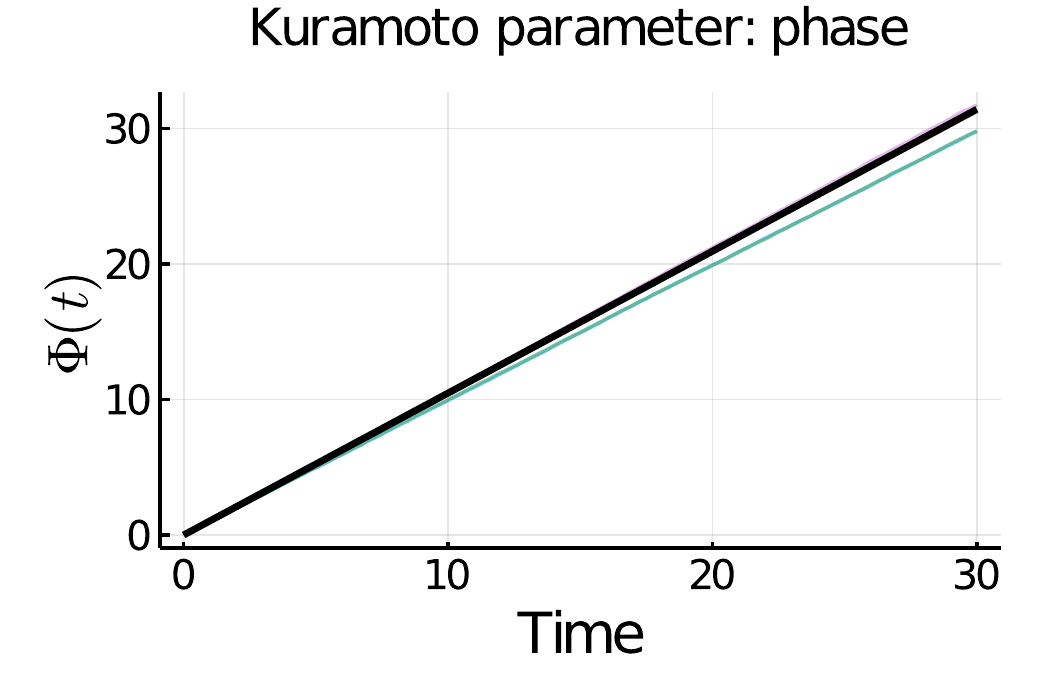}}
    \caption{
    In colour are five trajectories for $R(t)$ and $\Phi(t)$ for the high-fidelity model in \Cref{eq:Kuramoto_hi}. Stochasticity arises from sampling $\omega_i \sim \mathrm{Cauchy}(\omega_0, \gamma)$ for $i = 1,\dots, 256$. 
    In black is the deterministic trajectory from the low-fidelity model in \Cref{eq:Kuramoto_lo}.
    All simulations were completed using parameters $K=2$, $\omega_0 = \pi/3$ and $\gamma=0.1$.
    Note that the phase plot ``unwraps'' the trajectories of $\Phi(t) = \arg(Z_1(t))$ to avoid $2\pi$-discontinuities. 
    }
    \label{fig:eg_dynamics}
\end{figure}

Example trajectories of the high-fidelity and low-fidelity models in \Cref{eq:Kuramoto_hi,eq:Kuramoto_lo} are given in \Cref{fig:eg_dynamics}.
The trajectories $\obs R(t)$, $\obs \Phi(t)$, $R(t)$, $\Phi(t)$, $\tilde R(t)$ and $\tilde \Phi(t)$ on $t \in [0,30]$ are infinite dimensional.
In order to easily compare trajectories, we will select a finite number of informative \emph{summary statistics} from the trajectories, guided by the approximated system in \Cref{eq:Kuramoto_lo}.
We take
\begin{align*}
S_1(R, \Phi) &= \left( \frac{1}{30} \int_0^{30} R(t) ~\mathrm dt \right)^2, \\
S_2(R, \Phi) &= \frac{1}{30} \left( \Phi(30) - \Phi(0) \right), \\
S_3(R, \Phi) &= R \left( T_{1/2} \right),
\end{align*}
where $T_{1/2}$ is the first value of $t \in [0,30]$ for which $\obs R(t)$ is halfway between $\obs R(0)=1$ and its average value $S_1(\obs R, \obs \Phi)^{1/2}$.
Justification for the choice of these summary statistics and distances is provided in \Cref{appendix:summary_statistics}.

Simulation of the high-fidelity model produces $y \sim f(\cdot~|~(K, \omega_0, \gamma))$ by: 
(a) generating $\omega_i$, $i=1,\dots,256$, from $\mathrm{Cauchy}(\omega_0, \gamma)$; then 
(b) simulating the ODE system in \Cref{eq:Kuramoto_hi}; then 
(c) computing $y = (S_1(R, \Phi), S_2(R, \Phi), S_3(R, \Phi))$.
Simulation of the low-fidelity model produces $\tilde y \sim \tilde f(\cdot~|~(K, \omega_0, \gamma))$ by: 
(a) simulating the ODE system in \Cref{eq:Kuramoto_lo}; then
(b) computing $\tilde y = (S_1(\tilde R, \tilde \Phi), S_2(\tilde R, \tilde \Phi), S_3(\tilde R, \tilde \Phi))$.
The distances $d(y, \obs y)$ and $d(\tilde y, \obs y)$ are defined according to the weighted Euclidean norm, $d(a,b)^2 = 4(a_1 - b_1)^2 + (a_2 - b_2)^2 + (a_3 - b_3)^2$.
Note that the low-fidelity model is deterministic.
Therefore there is no meaningful definition of a coupling between the models at each fidelity: any simulation $y \sim f(\cdot~|~\tilde y,~(K, \omega_0, \gamma)) = f(\cdot~|~(K, \omega_0, \gamma))$ from the high-fidelity model will be independent of $\tilde y$.

\subsection{Existing ABC algorithms}
\label{s:existing}

We set independent uniform priors on $[1, 3]$, $[-2\pi, 2\pi]$ and $[0, 1]$ for $K$, $\omega_0$ and $\gamma$, respectively.
Using the uniform prior as importance distributions, samples from the ABC posterior, $p_{0.5}((K, \omega_0, \gamma) ~|~\obs y)$, are produced using \Cref{ABC:Importance} (ABC-RS), \Cref{ABC:SMC} (ABC-SMC), and \Cref{MFABC:Importance} (MF-ABC-RS) with $\epsilon = 0.5$.
The continuation probability used in \Cref{MFABC:Importance} (MF-ABC-RS) is the constant $\alpha \equiv 0.5$.
The resulting samples are depicted in \Cref{post:ABC:Rejection,post:ABC:SMC,post:MFABC:Rejection}.
\Cref{table:ABCFlavours} shows the observed values for the performance of each of these algorithms, quantified in terms of ESS, total simulation time, and observed efficiency (i.e. the ratio of the first two).
Note that the ESS of the sample from \Cref{ABC:SMC} (ABC-SMC) depends on only $N_4 = 1500$ weights, $w_n^{(4)}$, corresponding to the final generation.
However, we will measure the observed efficiency by using the total time to simulate, which includes the total simulation time of the preceding generations.

\begin{table}
\center
\begin{tabularx}{\textwidth}{l *4{Y}}
\toprule
 Algorithm & ESS & Sim. time (\si{\minute}) & \multicolumn{2}{c}{Efficiency (ESS \si{\per\minute})} \\
 \midrule
 \Cref{ABC:Importance} (ABC-RS) & 148.0 & 43.6 & 3.39 & ~ \\
 \Cref{ABC:SMC} (ABC-SMC) & 255.1 & 48.8 & 5.23 & $\times 1.54$ \\
 \Cref{MFABC:Importance} (MF-ABC-RS) & 126.4 & 22.9 & 5.52 & $\times 1.63$ \\
 \bottomrule
\end{tabularx}
\caption{Comparing existing ABC algorithms for sampling $p_{0.5}((K, \omega_0, \gamma)~|~\obs y)$ based on a uniform prior, using $\epsilon=0.5$. 
The stopping condition for \Cref{ABC:Importance,MFABC:Importance} is $N=6000$.
The threshold schedule for \Cref{ABC:SMC} is $(2, 1.5, 1, 0.5)$ with stopping conditions $S_t$ of $N_t = 1500$ parameter proposals for $t=1,2,3,4$, leading to the same total number of parameter proposals as \Cref{ABC:Importance,MFABC:Importance}. 
The perturbation kernels $K_t$ in \Cref{ABC:SMC} are Gaussian with diagonal covariance equal to twice the empirical variance of the sample at generation $t$~\cite{Beaumont2009}. 
The continuation probability used in \Cref{MFABC:Importance} is fixed at the constant $\alpha \equiv 0.5$.
Percentages refer to the increase in efficiency over the base efficiency of ABC-RS.}
\label{table:ABCFlavours}
\end{table}

Even with minimal tuning of \Cref{MFABC:Importance,ABC:SMC}, the samples built using these algorithms both show significant improvements in efficiency.
We have chosen stopping conditions to ensure equal number of parameter proposals for each algorithm, in order to demonstrate the distinct effects of each.
\Cref{ABC:SMC} (ABC-SMC) produces a larger ESS for a similar simulation time.
This is characteristic of ABC-SMC, whereby parameters with low likelihood are less likely to be proposed.
However, \Cref{MFABC:Importance} (MF-ABC-RS) instead speeds up the simulation time of the fixed number of parameter proposals, albeit with some damage to the ESS.
This result illustrates the orthogonal effects of the SMC and multifidelity ABC algorithms, and thus the potential for combining the techniques in \Cref{MFABC:SMC} to produce further gains in efficiency.

The key to the success of MF-ABC-RS and the multifidelity approach in general is the assumption that the low-fidelity model is cheaper to simulate than the high-fidelity model.
In this case, the high-fidelity model in \Cref{eq:Kuramoto_hi} has a mean (respectively, standard deviation) simulation time of approximately 520 (638) \si{\micro\second} to simulate, while the low-fidelity model in \Cref{eq:Kuramoto_lo} has a simulation time of approximately 10 (12) \si{\micro\second}.
Note that these averages and standard deviations are observed across the uniform distribution of parameter values on the intervals $[1, 3]$, $[-2\pi, 2\pi]$ and $[0, 1]$ for $K$, $\omega_0$ and $\gamma$, respectively.
The relatively large standard deviations imply that parameter proposals in this domain can produce very different simulation times.
Different importance distributions will therefore vastly alter the relative simulation costs of the high-fidelity and low-fidelity models.

\subsection{Multifidelity ABC-SMC}
\label{s:Results}

In order to demonstrate the increased efficiency of combining multifidelity approaches with SMC, we produced $100$ samples from the ABC posterior, $p_{0.1}((K, \omega_0, \gamma)~|~\obs y)$, consisting of $50$ replicates from each of \Cref{ABC:SMC} (ABC-SMC) and \Cref{MFABC:SMC} (MF-ABC-SMC).
Common to both algorithms is the number of generations, $T=8$, which corresponds to the nested sequence of ABC neighbourhoods $\Omega_{\epsilon_t}$ with the sequence of thresholds 2.0, 1.5, 1.0, 0.8, 0.6, 0.4, 0.2 and 0.1.
Each generation has a stopping condition of $\mathrm{ESS} \geq 400$, evaluated after every $100$ parameter proposals (to allow for parallelisation).
This condition reflects a specification that we need each generation's sample to be, in some sense, `good enough' to produce a reliable importance distribution that can be used in the next generation.
Finally, we specified the perturbation kernels $K_t(\cdot~|~(K, \omega_0, \gamma)^{(t)}_n)$ at each generation to be Gaussians centred on the parameter value $(K, \omega_0, \gamma)^{(t)}_n$.
The covariance matrices are diagonal matrices $\mathrm{diag}(\sigma_K^{(t)}, \sigma_{\omega_0}^{(t)}, \sigma_\gamma^{(t)})$, where 
\begin{align*}
(\sigma_K^{(t)})^2 &= 2 \frac{\sum |w_n^{(t)}| (K_n^{(t)} - \mu_K^{(t)})^2}{\sum |w_n^{(t)}|}, \\
\mu_K^{(t)} &= \frac{\sum |w_n^{(t)}| K_n^{(t)} }{\sum |w_n^{(t)}|},
\end{align*}
and similarly for $\sigma_{\omega_0}^{(t)}$ and $\sigma_\gamma^{(t)}$.
These perturbation kernels implement a typical choice for the covariance of using twice the empirical variance of the observed parameter values~\cite{Beaumont2009,Filippi2013}.
Note that we use this definition for the multifidelity case also, where weights $w_n^{(t)}$ may be negative, since we are using the $\rho_t$ (as defined in \Cref{eq:new_target}) and not the $p_{\epsilon_t}$ as the target distributions.
Further to these common inputs, the parameters $\rho_1$ and $\rho_2$ are the only additional algorithm parameters we need to specify to implement \Cref{MFABC:SMC} (MF-ABC-SMC).
We set lower bounds of $\rho_1 = \rho_2 = 0.01$ on the allowed continuation probabilities, with the aim to limit the variability of $w_n^{(t)}$ to prevent the collapse of the ESS.

\subsubsection{Multifidelity ABC-SMC increases observed efficiency}
\label{s:Results:Efficiency}

\Cref{ABC:SMC,MFABC:SMC} were implemented and run using Julia 1.5.1 on a 36 core CPU (2 $\times$ 18 core with hyperthreading), 2.3/3.7 GHz, 768 GB RAM.
\Cref{tab:results} quantifies the average performance for each of \Cref{ABC:SMC,MFABC:SMC}, separated out for each generation from $t=1,\dots,8$, and also across the entire SMC algorithm.
The final row of this table demonstrates that \Cref{MFABC:SMC} (MF-ABC-SMC) results in a 60\% saving in the total simulation time required to produce a final sample from the ABC posterior $p_{0.1}$ with an ESS of approximately 400.
This corresponds to an efficiency 2.48 times that of \Cref{ABC:SMC} (ABC-SMC).
This performance improvement is derived from an 85\% saving in the simulation time required for each parameter proposal when averaged across the entire SMC algorithm.

\begin{table}[]
    \centering
    \begin{tabularx}{\textwidth}{c *9{Y}}
        \toprule
        ~ 
            & \multicolumn{3}{c}{a. Sim. time / proposal} 
            & \multicolumn{3}{c}{b. Simulation time} 
            & \multicolumn{3}{c}{c. Efficiency} \\
        ~
            & \multicolumn{3}{c}{\si{\micro\second}} 
            & \multicolumn{3}{c}{\si{\minute}} 
            & \multicolumn{3}{c}{ESS \si{\per\second}} \\
        \cmidrule(lr){2-4} \cmidrule(lr){5-7} \cmidrule(l){8-10}
        $t$ & ABC & \multicolumn{2}{c}{MF-ABC}  & ABC & \multicolumn{2}{c}{MF-ABC}  & ABC & \multicolumn{2}{c}{MF-ABC} \\
        \midrule
        1 & 537 & 530 & $-1\%$ & 13.6 & 13.4 & $-2\%$ & 0.51 & 0.52 & $\times 1.01$ \\
        2 & 675 & 31 & $-95\%$ & 11.2 & 1.3 & $-88\%$ & 0.63 & 8.35 & $\times 13.4$ \\
        3 & 665 & 33 & $-95\%$ & 15.2 & 2.0 & $-87\%$ & 0.46 & 4.41 & $\times 9.59$ \\
        4 & 599 & 35 & $-94\%$ & 11.6 & 1.9 & $-84\%$ & 0.60 & 4.99 & $\times 8.26$ \\
        5 & 525 & 35 & $-93\%$ & 11.3 & 2.2 & $-81\%$ & 0.62 & 3.99 & $\times 6.47$ \\
        6 & 424 & 36 & $-92\%$ & 10.8 & 2.6 & $-76\%$ & 0.64 & 2.90 & $\times 4.53$ \\
        7 & 330 & 46 & $-86\%$ & 13.6 & 6.0 & $-56\%$ & 0.50 & 1.16 & $\times 2.31$ \\
        8 & 268 & 80 & $-70\%$ & 15.5 & 11.8 & $-24\%$ & 0.43 & 0.57 & $\times 1.31$ \\
        \midrule
        SMC & 449 & 65 & $-85\%$ & 102.9 & 41.2 & $-60\%$ & 0.066 & 0.163 & $\times 2.48$ \\ 
        \bottomrule
    \end{tabularx}
    \caption{
    Row $t$ shows empirical mean values for 50 ABC-SMC samples and 50 MF-ABC-SMC samples of:
    (a) the simulation time divided by the number, $N_t$, of parameter proposals in generation $t$;
    (b) the simulation time in generation $t$;
    (c) ESS divided by simulation time in generation $t$.
    The final row is the empirical mean values for each sample, of
    (a) the total simulation time divided by the total number, $\sum_t N_t$, of parameter proposals;
    (b) the total simulation time;
    (c) ESS from generation $8$ divided by total simulation time (all generations).
    The third columns for each of (a), (b) and (c) quantify the improvement in performance of MF-ABC-SMC over ABC-SMC.
    }
    \label{tab:results}
\end{table}

\Cref{fig:empirical_means} depicts 100 estimates of the posterior mean of each of the parameters, $\mathbb E_{p_{0.1}}(K)$, $\mathbb E_{p_{0.1}}(\omega_0)$ and $\mathbb E_{p_{0.1}}(\gamma)$, constructed from the 50 samples generated by each of \Cref{ABC:SMC,MFABC:SMC}.
This figure demonstrates that the estimates generated by the multifidelity algorithm, MF-ABC-SMC, are of a similar quality to those produced by ABC-SMC.\footnote{
In addition to the observation of the posterior means in \Cref{fig:empirical_means}, we have also depicted representative posterior samples from each algorithm in \Cref{post:ABC:SMC:ESS400,post:MFABC:SMC:ESS400}.
}
Taking $K$, $\omega_0$ and $\gamma$ in turn, the means of the 50 estimates produced by each algorithm are, to three significant figures, indistinguishable at 2.17, 1.06 and 0.125.
Similarly, the variability of these estimates are also broadly similar: \Cref{ABC:SMC} (respectively, \Cref{MFABC:SMC}) produces estimates with standard deviations .0268 (.0277), .0034 (.0027), and .00271 (.00265).
Importantly, the distribution of total simulation times for each of these 100 estimates demonstrates that MF-ABC-SMC is reliably significantly less computationally expensive than ABC-SMC to produce comparable samples, with the average simulation times reflecting the 60\% saving identified in \Cref{tab:results}.

\begin{figure*}
\centering
\includegraphics[width=0.9\textwidth]{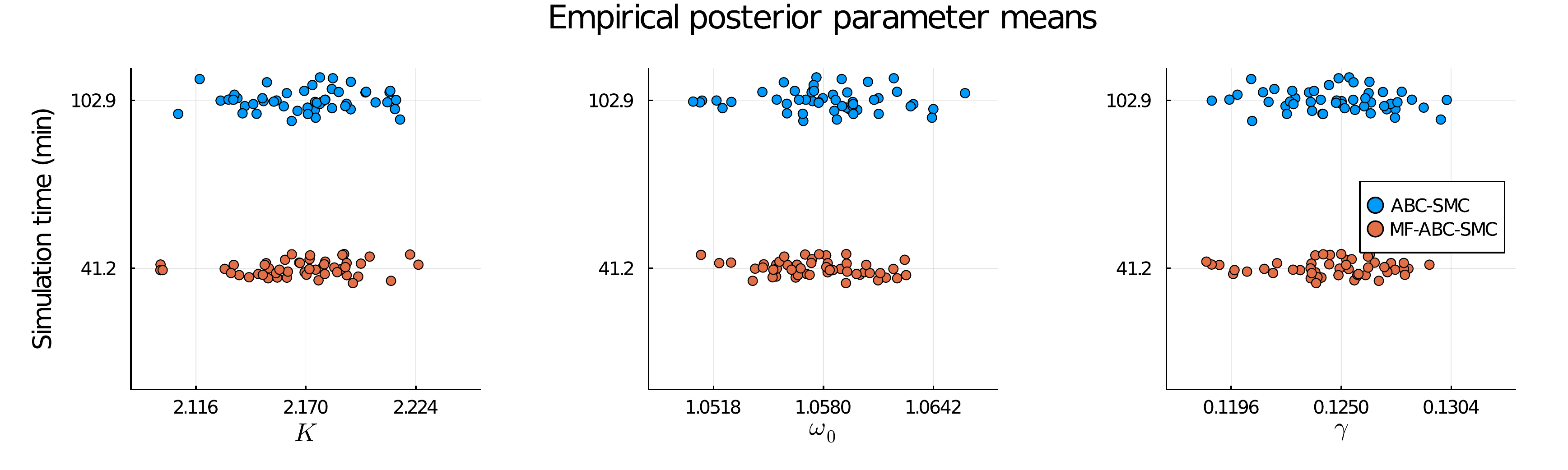}
\caption{
The empirical posterior means of each of $K$, $\omega_0$ and $\gamma$ for each of the 100 samples, plotted against the total simulation time required to generate each mean.
}
\label{fig:empirical_means}
\end{figure*}

Recall our initial observation that the average simulation time of the low-fidelity model, 10 \si{\micro\second}, is approximately 2\% of that of the high-fidelity model, 520 \si{\micro\second}, with expectations taken over the prior.
Given this initial difference in each model's average simulation time, the observed 148\% increase in efficiency from \Cref{ABC:SMC} to \Cref{MFABC:SMC} is determined by a number of other factors specific to the characteristics of SMC sampling, which we now explore.

\subsubsection{MF-ABC-SMC is more effective in early generations}
\label{s:Results:Generations} 

In producing the samples summarised in \Cref{tab:results,fig:empirical_means}, we enforced the same decreasing schedule of $\epsilon_t$ and the same stopping criteria ($\mathrm{ESS} \geq 400$) for all runs of both \Cref{ABC:SMC} (ABC-SMC) and \Cref{MFABC:SMC} (MF-ABC-SMC).
This allows a direct comparison between the two algorithms of the efficiencies at each generation, in addition to their overall performance.
\Cref{fig:efficiencies_generation} depicts the distributions of each of the measures for which the means are given in \Cref{tab:results}.
As the generation index varies, there are significant differences between the performance improvement generated by MF-ABC-SMC over ABC-SMC.
By all three of the measures, the benefit of MF-ABC-SMC appears to accrue most significantly in the earlier generations.
There are a number of factors that explain these observed differences.

\begin{figure}
\centering
\includegraphics[width=0.9\textwidth]{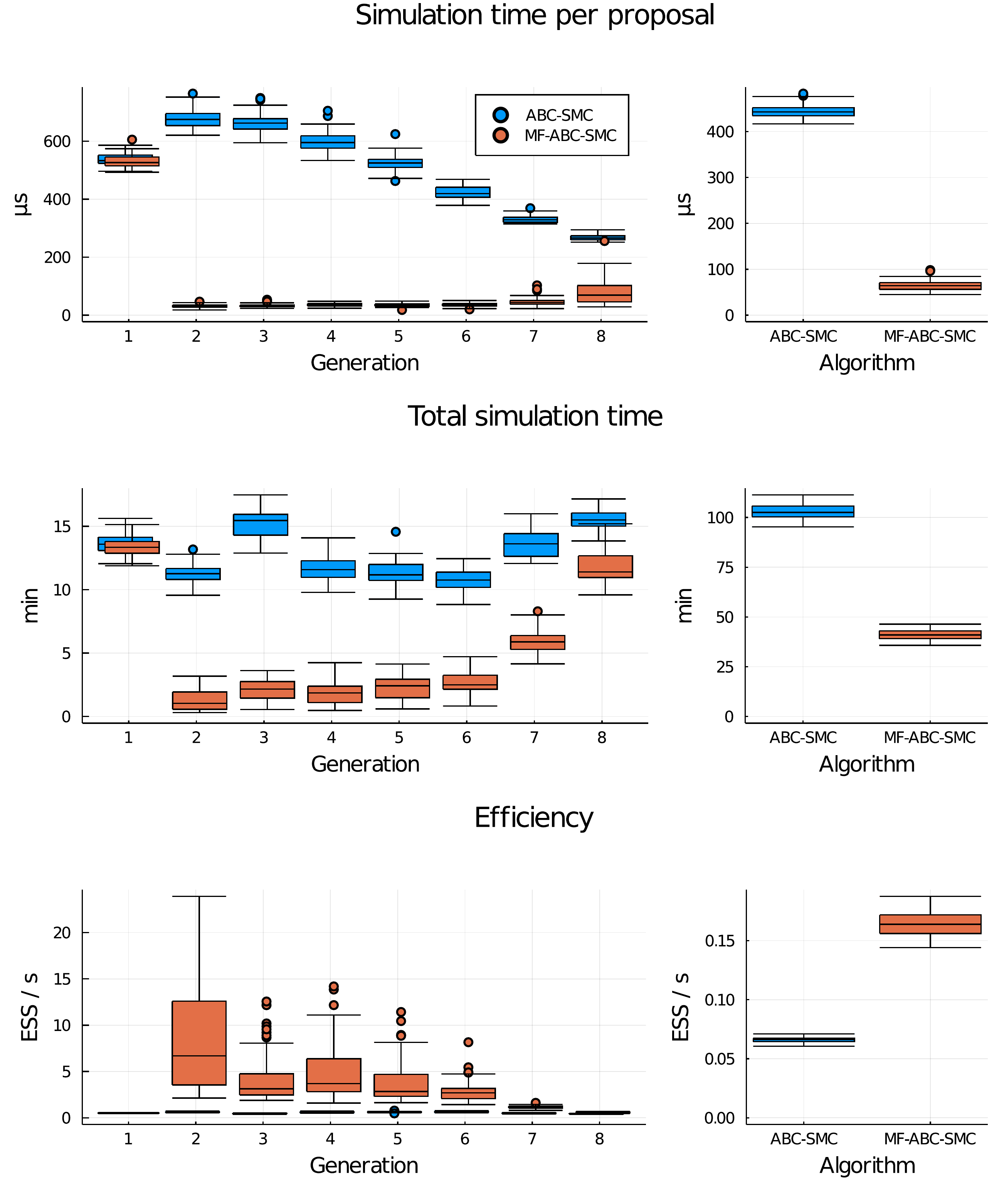}
\caption{
Observed distributions of simulation time per proposal, total simulation time, and efficiency for 50 samples from \Cref{ABC:SMC} (ABC-SMC, blue) and 50 samples from \Cref{MFABC:SMC} (MF-ABC-SMC, orange). 
Left: Summarised at each SMC generation. 
Right: Aggregated across all eight generations.
Means of each observed distribution are given in \Cref{tab:results}.
}
\label{fig:efficiencies_generation}
\end{figure}

First, we observe that the average simulation time per proposal of \Cref{ABC:SMC} (ABC-SMC) decreases as $t$ increases.
As the importance distribution evolves towards the posterior through the SMC algorithm, the most expensive high-fidelity simulations are required much less often.
This means that the relative saving of using the low-fidelity model also evolves with the generation, $t$.
The continuation probabilities used in \Cref{MFABC:SMC} (MF-ABC-SMC) aim to balance the saving in simulation cost against the probability of false positives and false negatives, according to the optima in \Cref{eq:etastar:unbounded}.
Since there is less computational cost available to save in later generations, we therefore find larger continuation probabilities.
\Cref{fig:etas} shows how the values of $\eta_1$ and $\eta_2$ vary for each of the 50 instances of \Cref{MFABC:SMC} (MF-ABC-SMC) across generations $t=2,4,6,8$, reflecting the adaptation of the algorithm to changing distributions in simulation costs.
Values in the bottom-left of the figure correspond to smaller continuation probabilities, and thus smaller total simulation times.
At $t=8$, the optimal continuation probabilities are no longer clustered to the bottom-left of the figure, but instead have begun to migrate to the $(1,1)$ corner, which corresponds to the classical ABC-SMC approach.
This leads to the increase in simulation time per proposal for MF-ABC-SMC at generation $t=8$, as more high-fidelity simulations are required.

\begin{figure}
\centering
\includegraphics[width=0.9\textwidth]{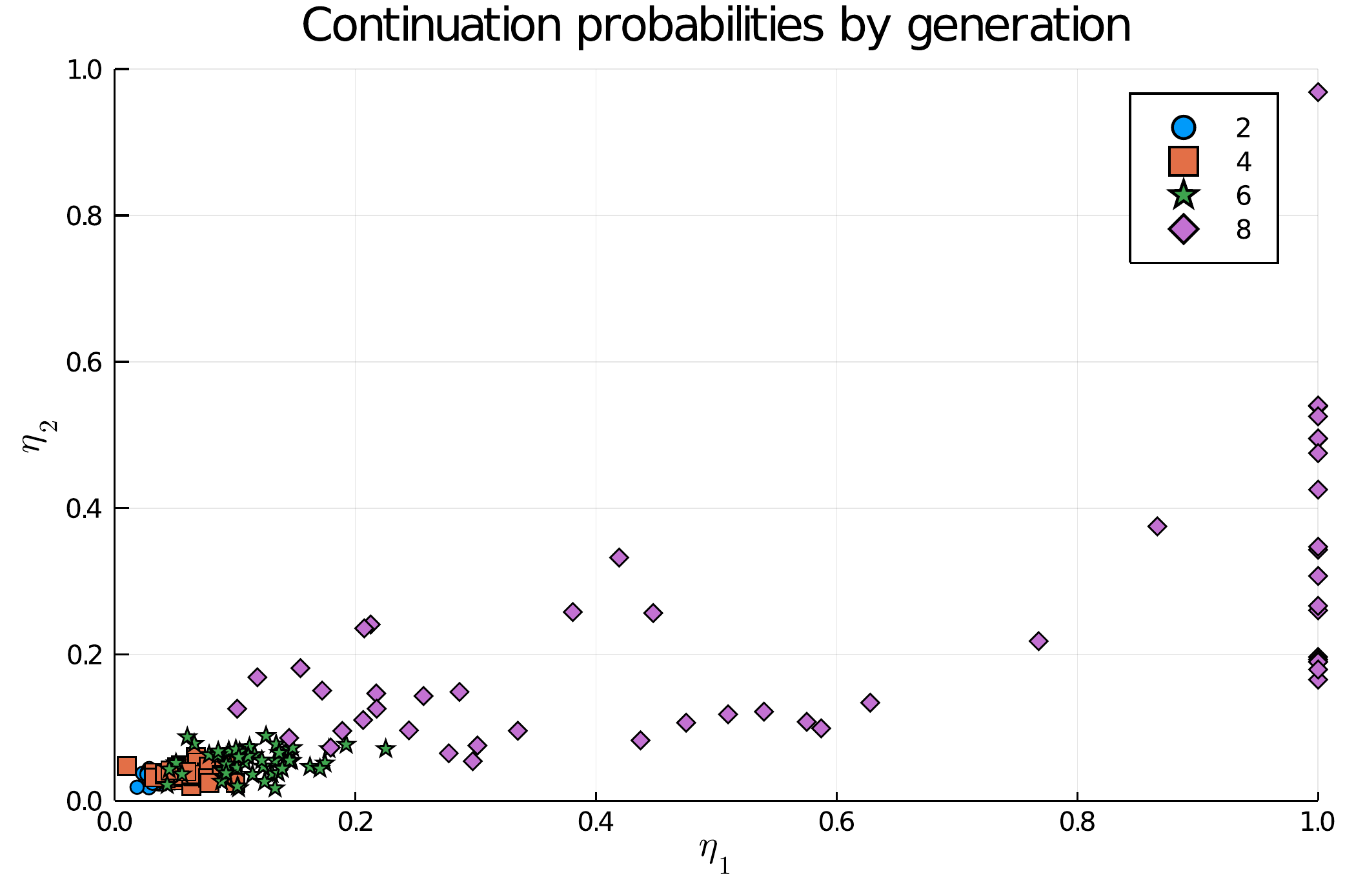}
\caption{
Estimated values of the optimal continuation probabilities $\alpha = \eta_1 \mathbb I(\tilde y \in \Omega_{\epsilon_t}) + \eta_2 \mathbb I(\tilde y \notin \Omega_{\epsilon_t})$ used to produce the multifidelity sample for each of the 50 replicates of generations 2, 4, 6 and 8.
}
\label{fig:etas}
\end{figure}

Second, we observe that for generation $t=8$ we have an average of 70\% improvement in simulation time per proposal but only a 24\% average improvement in the generation's total simulation time.
It follows that MF-ABC-SMC requires approximately 2.5 times as many parameter proposals, on average, in the final generation to produce an ESS of at least 400 when compared with ABC-SMC, meaning that only a 30\% improvement in efficiency can be found.
Thus, the final generation incurs over a quarter of the total simulation cost of MF-ABC-SMC, on average.
The improved efficiency at the final generation is relatively small because at smaller ABC thresholds the inaccuracy of the low-fidelity model becomes more significant.
In this example, the underlying stochasticity of the high-fidelity model means that the probability $\mathbb P(y \in \Omega_{0.1}~|~ \tilde y \in \Omega_{0.1})$ is sufficiently small to produce a high false positive rate, corresponding to a relatively large value of $\fp W$ in \Cref{etastar:unbounded}.
This reduces the optimal available efficiency towards that of classical ABC-SMC.
It should be noted that a strength of the adaptive method of selecting continuation probabilities is that this case can be detected, so that the efficiency of \Cref{MFABC:SMC} (MF-ABC-SMC) is bounded below by the efficiency of \Cref{ABC:SMC} (ABC-SMC).

Finally, note that the sequence of $\epsilon_t$ was chosen to be equal across ABC-SMC and MF-ABC-SMC to allow for a direct comparison between the efficiencies at each generation.
This constraint produced an overall 60\% saving in simulation time, and a 148\% increase in efficiency.
However, in practice, we should also aim to adaptively choose the sequence $\epsilon_t$ (for either algorithm) to optimise performance.
In the following subsection, we describe how incorporating the adaptive selection of $\epsilon_t$ based on the preceding generations' output allows for a better comparison between the two algorithms.

\subsubsection{MF-ABC-SMC reduces ABC approximation bias}
\label{s:Adaptive}

Implementing \Cref{MFABC:SMC} (MF-ABC-SMC) requires a decreasing sequence of $\epsilon_t$ values to be pre-specified.
Often, appropriate values of $\epsilon_t$ cannot be determined before any simulations have been generated.
If the sequence of $\epsilon_t$ decreases too slowly then the algorithm will take a long time to reach the target posterior; too quickly, and acceptance rates will be too low.
As a result, rather than specifying a sequence of thresholds \emph{a priori}, previous work in the SMC context~\cite{DelMoral2012} has explored choosing $\epsilon_{t+1}$ by predicting its effect on the ESS of that generation.
In the spirit of this approach, we adaptively choose the sequence $\epsilon_t$ by predicting its effect on the efficiency of that generation, as defined in \Cref{lemma:phi}.

\begin{figure}
    \centering
    \subfloat[Adaptive ABC-SMC.\label{fig:adaptive_epsilon}]{\includegraphics[width=0.45\textwidth]{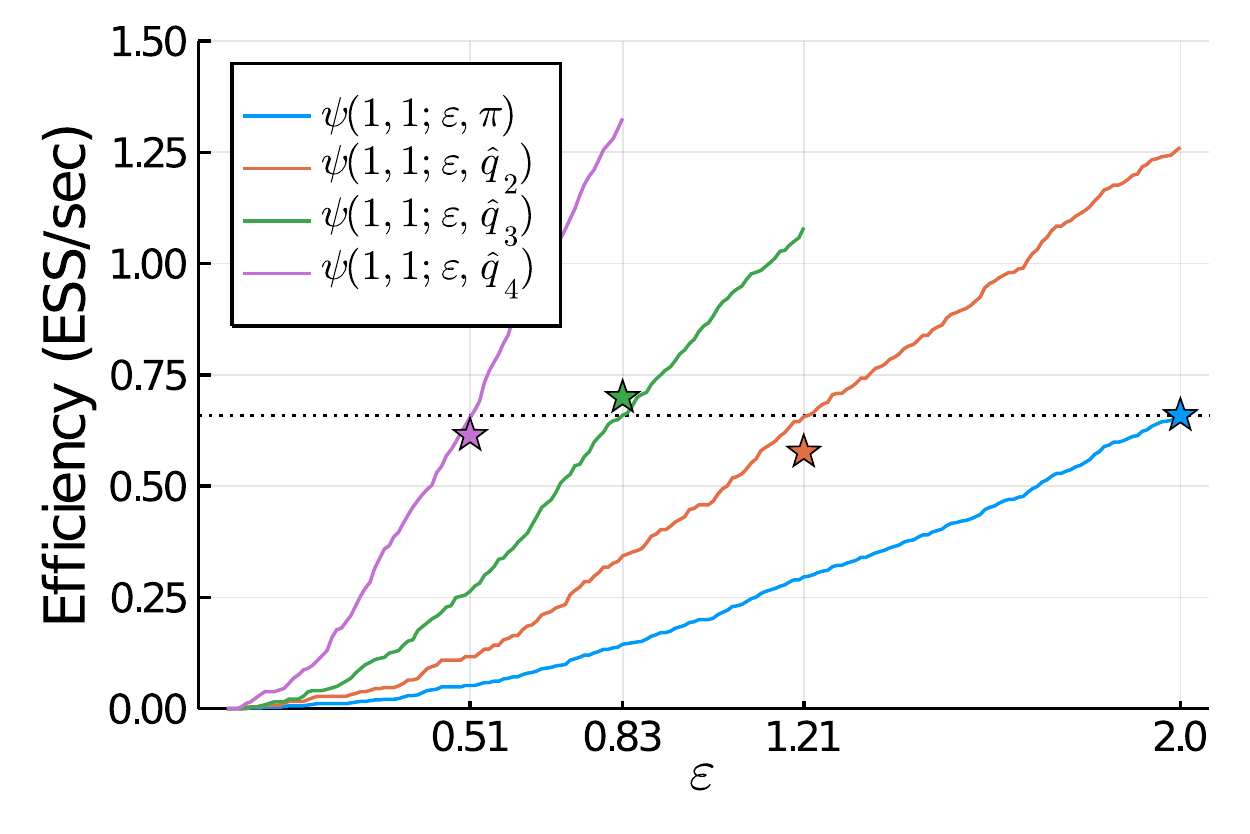}}~
    \subfloat[Adaptive MF-ABC-SMC.\label{fig:adaptive_epsilon_eta}]{\includegraphics[width=0.45\textwidth]{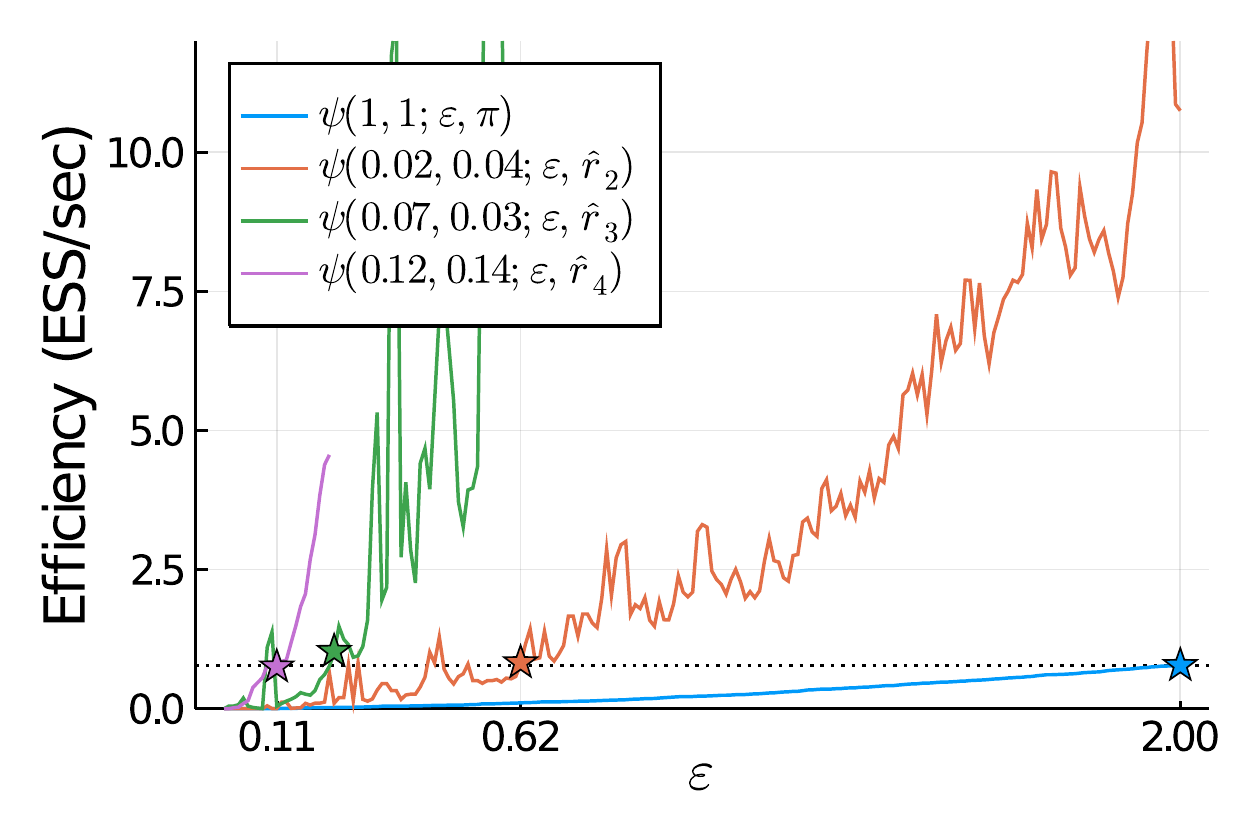}} 
    \caption{
    Adaptively selecting $\epsilon_t$ to achieve a predicted target efficiency equal to the efficiency of the first generation. 
    Curves plot the estimated efficiency as a function of $\epsilon$ for (a) importance distributions $\hat q_t$, and (b) importance distributions $\hat r_t$ and associated optimised continuation probabilities.
    Stars plot observed efficiencies, $\mathrm{ESS}/T_{\mathrm{tot}}$, against the selected value of $\epsilon_t$.
    }
    \label{fig:adaptive}
\end{figure}

\Cref{fig:adaptive} depicts a possible strategy for choosing $\epsilon_{t+1}$ conditionally on the output from generation $t$ in ABC-SMC and MF-ABC-SMC.
The key to this strategy is the observation that the efficiency of \Cref{MFABC:Importance}, defined by \Cref{eq:Phi,eq:PhiComponents}, depends on the ABC threshold, $\epsilon$, and the importance distribution, $\hat q$, used in \Cref{MFABC:Importance}.
By writing the efficiency as $\psi(\eta_1, \eta_2;~\epsilon, \hat q)$ and fixing $\eta_1$, $\eta_2$ and $\hat q$, we can consider the efficiency of \Cref{MFABC:Importance} as a function of $\epsilon$.
In particular, by setting $\eta_1 = \eta_2 = 1$, the efficiency of \Cref{ABC:Importance}, $\psi(1, 1; \epsilon, \hat q)$, also varies with $\epsilon$.
We assume that at each generation we have a target efficiency, $\psi_t^\star$, that enables a given ESS to be generated with a known computational budget.

In generation $t$ of \Cref{MFABC:SMC} (MF-ABC-SMC), we produce the sample $\{ \theta_n^{(t)}, w_n^{(t)} \}$ from \Cref{MFABC:Importance} (MF-ABC-IS), which can used to define an importance distribution, $\hat r_{t+1}$.
Steps 5 and 6 of \Cref{MFABC:SMC} then use the next ABC threshold, $\epsilon_{t+1}$, to calculate optimal continuation probabilities $(\eta_1^\star, \eta_2^\star)$ by maximising the efficiency function, $\psi(\eta_1, \eta_2; \epsilon_{t+1}, \hat r_{t+1})$.
In the case where $\epsilon_{t+1}$ is unknown, an adaptive approach to finding an appropriate value is to replace steps 5 and 6 of \Cref{MFABC:SMC} with the following subroutine:
\begin{enumerate}[label=\alph*.]
    \item find $(\eta_1^\star, \eta_2^\star)$ to maximise the efficiency, $\psi(\eta_1, \eta_2; \epsilon_t, \hat r_{t+1})$, using \Cref{etastar};
    \item set $\epsilon^{\star} \leftarrow \max \left\{ 0 < \epsilon \leq \epsilon_t ~:~ \psi(\eta_1^\star, \eta_2^\star; \epsilon, \hat r_{t+1}) \leq \psi^\star_{t+1} \right\} $, or $\epsilon^{\star} \leftarrow \epsilon_t$ if this set is empty;
    \item set $\epsilon_{t+1} \leftarrow \epsilon^\star$ and continue to step 7 of \Cref{MFABC:SMC}.
\end{enumerate}
This procedure produces a sequence of ABC thresholds designed such that each generation's efficiency is maintained at a target level.
It alternates between finding continuation probabilities that maximise the efficiency at the preceding threshold, $\epsilon_t$, and then finding a value of $\epsilon_{t+1} < \epsilon_t$ such that the predicted efficiency matches the target.
If $\epsilon_{t+1}=\epsilon_t$ then the target efficiency, $\psi_{t+1}^\star$, needs to be reviewed as it is not achievable.
Note that for the case of ABC-SMC, which is equivalent to MF-ABC-SMC with fixed continuation probabilities $\eta_1 = \eta_2 = 1$, we skip the optimisation in step a and use $\eta_1^\star = \eta_2^\star = 1$ in step b.

The strength of the multifidelity approach in this context is that the additional degrees of freedom afforded by the continuation probabilities allows for the ABC thresholds to decrease more quickly, while maintaining a target efficiency at each generation.
This benefit is depicted in \Cref{fig:adaptive}.
Each curve in the left-hand plot is the predicted efficiency $\psi(1, 1; \epsilon, \hat q_{t+1})$ as a function of $\epsilon$.
Each curve in the right-hand plot is the predicted efficiency $\psi(\eta_1^\star, \eta_2^\star; \epsilon, \hat r_{t+1})$, again as a function of $\epsilon$, where we have found optimal continuation probabilities, $(\eta_1^\star, \eta_2^\star)$.
For each algorithm, we choose the target efficiency at each generation $t>1$ to equal the efficiency observed in generation $1$.
\Cref{fig:adaptive_epsilon} demonstrates the adaptive threshold selection implemented for four generations of \Cref{ABC:SMC} (ABC-SMC), producing a decreasing sequence for $\epsilon_t$ of $2 > 1.21 > 0.83 > 0.51$.
In comparison, \Cref{fig:adaptive_epsilon_eta} shows how the adaptive selection of thresholds in \Cref{MFABC:SMC} (MF-ABC-SMC) produces a sequence for $\epsilon_t$ of $2 > 0.62 > 0.23 > 0.11$.
Clearly, the adaptive sequence of $\epsilon_t$ enabled by MF-ABC-SMC decreases much more quickly than the equivalent sequence for ABC-SMC, while the efficiency of each generation remains broadly constant and predictable.

In our example, each generation has a stopping condition of $\mathrm{ESS} \geq 400$.
As a result, by choosing to specify a constant target efficiency equal to the observed efficiency of generation $t=1$, we effectively impose a constant target simulation budget for each generation.
Since (for the example in \Cref{fig:adaptive}) we have specified four generations for each run of the adaptive versions of \Cref{ABC:SMC,MFABC:SMC}, we have thus specified a fixed, equal total simulation time for each algorithm.
In this setting, the results in \Cref{fig:adaptive} show that the bias incurred by using an ABC approximation to the posterior with threshold $\epsilon>0$ is vastly reduced by implementing the multifidelity approach to SMC.
Using MF-ABC-SMC, the sample produced in generation $4$ is from $p_{0.11}(\theta~|~\obs y)$, while ABC-SMC can only produce a sample from $p_{0.51}(\theta~|~\obs y)$ for a similar computational cost.
Thus, by incorporating the multifidelity approach into a method for adaptively selecting ABC thresholds, we can afford to allow the sequence $\epsilon_t$ of MF-ABC-SMC thresholds to decrease much more rapidly, at no cost to the efficiency of the algorithm.

\section{Discussion and conclusions}
\label{s:Discussion}

In this work we have examined how to integrate two approaches to overcoming the computational bottleneck of repeated simulation within ABC: 
the SMC technique for producing parameter proposals,
and the multifidelity technique for using low-fidelity models to reduce overall simulation time.
By combining these approaches, we have produced the MF-ABC-SMC algorithm in \Cref{MFABC:SMC}.

The results in \Cref{s:Example} demonstrate that the efficiency of sampling from the ABC posterior (measured as the ratio of the ESS to simulation time) can be significantly improved by using \Cref{MFABC:SMC} (MF-ABC-SMC) in place of \Cref{ABC:SMC} (ABC-SMC).
This improvement was demonstrated by using a common schedule of decreasing ABC thresholds for both algorithms.
In this case, the increase in efficiency was most significant during the early SMC generations, where both the average simulation time of the high fidelity model and overall acceptance rates are relatively large.
By also implementing an adaptive ABC thresholding scheme into both algorithms and thus allowing different sequences of ABC threshold, MF-ABC-SMC is shown to greatly reduce the bias incurred by using an ABC approximation to the posterior, in comparison to ABC-SMC.

Having introduced the MF-ABC-SMC algorithm, a number of open questions emerge.
Some of these questions are specific to the implementation of multifidelity approaches.
However, others arise from a re-evaluation of SMC implementation strategies in this new context.
Below, we consider these two classes of question in turn.

\subsection{Multifidelity implementation}

The key to maximising the benefit of MF-ABC-SMC is the ability to set a continuation probability based on the simulations generated during preceding generations.
The estimates in \Cref{eq:MonteCarlo} of the quantities in \Cref{eq:PhiComponents} are natural Monte Carlo approximations to the required integrals.
In the SMC context, when generating the continuation probability for generation $t+1$, each of these estimates could actually be constructed using the parameter proposals, importance weights, simulations, distances, and continuation probabilities of any (or all) generations $1 \leq s \leq t$, not just generation $t$.
Future work should clarify how best to combine many generations' samples into estimates of \Cref{eq:PhiComponents},
and the potential for improvement that might arise from this.

Another question arises when we break the assumption made at the start of \Cref{s:MFABC}, and no longer assume that the output spaces of each model fidelity are such that $\tilde{\mathcal Y} = \mathcal Y$.
In general, the observed data, $\obs{\tilde y} \neq \obs y$, the distance metrics, $\tilde d(\tilde y, \obs{\tilde y}) \neq d(y,\obs y)$, and the thresholds, $\tilde \epsilon \neq \epsilon$, may all be distinct.
In this case, any estimate, $\tilde w$, of $\mathbb I(y \in \Omega_\epsilon)$ can be used in place of $\mathbb I(\tilde y \in \Omega_\epsilon)$ to give a multifidelity acceptance weight of the form
\[
 w(\theta, \tilde y, u, y) = \tilde w(\theta, \tilde y) + \frac{\mathbb I(u<\alpha(\theta, \tilde y))}{\alpha(\theta, \tilde y)} \left( \mathbb I(y \in \Omega_\epsilon) - \tilde w(\theta, \tilde y) \right).
\]
Note that, in the case of equal output spaces, we might consider using $\tilde w(y_n) = \mathbb I(\tilde y \in \Omega_{\tilde \epsilon})$ for distinct thresholds, $\tilde \epsilon \neq \epsilon$.
However, in general, $\tilde w(\theta, \tilde y)$ may also encompass completely different output spaces based on distinct modelling frameworks for the same system (albeit with the same parameter space).
In a similar way to evolving acceptance probabilities across generations, we could also evolve the estimate $\tilde w$ across generations, using the information gathered from repeated simulation of both high-fidelity and low-fidelity models to better approximate $\mathbb I(y \in \Omega_\epsilon)$ and thus reduce our reliance on the high-fidelity model.

The form of continuation probability, $\alpha(\theta, \tilde y)$, defined in \Cref{eq:constantrates} implements a multifidelity algorithm that provides a single continuation probability for $\tilde y \in \Omega_\epsilon$ and another for $\tilde y \notin \Omega_\epsilon$, independently of the parameter value.
There may be significant improvements to the multifidelity approach available through making $\alpha(\theta, \tilde y)$ depend more generally on $\theta$ and $\tilde y$.
For example, it is probable that there should be less need to simulate $y$ after generating $\tilde y$ such that $d(\tilde y, \obs y) \gg \epsilon$, than if $d(\tilde y, \obs y) = 1.01 \epsilon$.
However, with $\alpha$ as defined in \Cref{eq:constantrates}, these two cases are treated equally.
There is likely to be significant potential for improved performance from exploring less simplistic forms for the continuation probability.

Finally, as has been noted in previous work on multifidelity ABC~\cite{Prescott2020}, there is much potential in being able to use multiple low-fidelity models, beyond a single low-fidelity approximation.
If there exist multiple low-fidelity models, different generations of MF-ABC-SMC may allow us to progressively focus on using the most efficient model, and identify the specific regions of parameter space for which one model or another may bring most benefit for parameter estimation.

\subsection{SMC implementation}

Previous research into the implementation of ABC-SMC has ensured that the importance distribution formed from the preceding generation, as given in \Cref{eq:importance}, is optimal, by choice of the perturbation kernels, $K_t$~\cite{Filippi2013}. 
This has typically been treated as a requirement to trade-off a wide exploration of parameter space against a high acceptance rate.
We have replaced the acceptance rate by the theoretical efficiency as the quantification of an ABC algorithm's performance.
Therefore, since we now explicitly include the simulation time in the definition of the algorithm's performance, the perturbation kernels that optimise the tradeoff between efficiency and exploration may be reformulated and hence different optima used.

In \Cref{s:Results} we applied a widely-used strategy for determining the perturbation kernels \cite{Beaumont2009}.
This strategy has been justified only in the context of positive weights and an importance distribution that approximates the ABC posterior.
However, the importance distribution used in MF-ABC-SMC potentially makes this choice of perturbation kernels suboptimal.
It remains to extend existing results on the optimality of perturbation kernels, such as those in \cite{Filippi2013}, to apply to importance distributions of the form in \Cref{eq:new_importance} that approximate the alternative target distribution in \Cref{eq:new_target}.
There is therefore justification for reopening the question of specifying optimal perturbation kernels for SMC, in the context of both including simulation time in the performance tradeoff and for dealing with multifidelity samples with negative weights.

We have restricted the choice of SMC sampler to the $O(N^2)$ PMC sampler, due to the effect of the $O(N)$ sampler in diluting the benefits of the multifidelity approach, as discussed in \Cref{s:MF-ABC-SMC,LinearSMCSampler}.
Although this choice can be justified when simulation times dominate the algorithm, future work in this area should seek to optimise a multifidelity approach in the context of the more efficient sampling technique.
Given that the observed effect of multifidelity ABC is to significantly reduce the simulation time per parameter proposal, and thereby make the simulation cost at each iteration much \emph{less} dominant, this extension will be necessary to ensure optimal performance of MF-ABC-SMC in larger sample sizes.

In \Cref{s:Adaptive} we described how to adapt \Cref{ABC:SMC,MFABC:SMC} to implement an adaptive sequence of $\epsilon_t$ as an approach to minimising bias for a fixed computational budget.
The strategy we used was to choose $\epsilon_t$ to maintain an efficiency as close as possible to a target, set to equal the observed efficiency of the first generation.
Further work in this area should investigate the use of more sophisticated strategies for choosing each $\epsilon_t$.
This question relates closely to the sequence of stopping criteria.
In \Cref{s:Results}, we constrained the effective sample size at each generation to be at least $400$ to ensure a relatively low variance in each generation's sample, but this choice was made arbitrarily.
Future work should therefore consider how to best achieve the ultimate goal of the SMC algorithm: a sample from a final generation with minimal bias relative to the true posterior, with a small variance, and constructed quickly.
This goal should be achieved by optimising the complex, interdependent choices of stopping criteria, ABC thresholds, continuation probabilities and perturbation kernels.


\newpage
\appendix

\section{Linear SMC Sampler}
\label{LinearSMCSampler}
The following procedure briefly describes the linear SMC sampling method of Del Moral et al, (2012) for the multifidelity context.
For each $z_n^{(t)} = (\theta_n^{(t)}, \tilde y_n^{(t)}, u_n^{(t)}, y_n^{(t)})$ produced in generation $t$, with weight $W_n^{(t)}$, a perturbation to $z_{\star}$ is proposed with density 
\[ g(z_\star) = K_t(\theta_\star~|~\theta_n^{(t)}) \check f(\tilde y_\star, y_\star~|~\theta_\star) \]
on $\mathcal Z = \Theta \times \mathcal Y \times [0,1] \times \mathcal Y$.
The proposal is accepted with a Metropolis--Hastings acceptance probability based on the ratio
\[
\frac{\hat w_t(z_\star)}{\hat w_t(z_n^{(t)})} 
\frac{K_t(\theta_n^{(t)}~|~\theta_\star) \pi(\theta_\star)}{K_t(\theta_\star~|~\theta_n^{(t)}) \pi(\theta_n^{(t)})},
\]
for the non-negative multifidelity weight
\[
\hat w_t(z) = \left| \mathbb I(\tilde y \in \Omega_{\epsilon_t}) + \frac{\mathbb I(u < \alpha_t)}{\alpha_t} \left( \mathbb I(y \in \Omega_{\epsilon_t}) - \mathbb I(\tilde y \in \Omega_{\epsilon_t})  \right) \right|.
\]
The MCMC step thus produces a sample point for the new generation, $z_n^{(t+1)}$, based on the MH acceptance step.
In the subsequent generation, the weight $W_n^{(t)}$ is then updated such that
\begin{equation}
	\label{eq:w_SMC}
	W_n^{(t+1)} \propto W_n^{(t)} \frac{\hat w_{t+1}(z_n^{(t+1)})}{\hat w_t(z_n^{(t+1)})}.
\end{equation}
Thus, for each $z_n^{(t+1)}$ we need to calculate both $\hat w_{t}(z_n^{(t+1)})$ and $\hat w_{t+1}(z_n^{(t+1)})$.

The benefit of the multifidelity weight is that, for $u < \alpha_t$, we only need to generate $\tilde y$ rather than $(\tilde y, y)$, avoiding a significant computational cost.
However, the weight update step given in \Cref{eq:w_SMC} requires two weight calculations dependent on $z_n^{(t+1)}$, each of which may require the expensive computation of $(\tilde y, y) \sim \check f(\cdot,\cdot~|~\theta)$.
Indeed, the cost saving of needing to generate $\tilde y \sim \tilde f(\cdot~|~\theta)$ alone for the weight update step exists only if
\[
u_n^{(t+1)} > \max \left\{
\alpha_t \left( \theta_n^{(t+1)}, \tilde y_n^{(t+1)} \right), 
\alpha_{t+1} \left( \theta_n^{(t+1)}, \tilde y_n^{(t+1)} \right) 
\right\}.
\]
Using the linear SMC sampling procedure thereby increases the simulation cost of updating the weights, relative to the $O(N^2)$ sequential importance sampling approach.
We assume that simulation costs dominate the cost of calculating importance weights, and thus we focus on sequential importance sampling.

\section{Choice of summary statistics}
\label{appendix:summary_statistics}

Recall the summary statistics
\begin{align*}
	S_1(R, \Phi) &= \left( \frac{1}{30} \int_0^{30} R(t) ~\mathrm dt \right)^2, \\
	S_2(R, \Phi) &= \frac{1}{30} \left( \Phi(30) - \Phi(0) \right), \\
	S_3(R, \Phi) &= R \left( T_{1/2} \right),
\end{align*}
where $T_{1/2}$ is the first value of $t \in [0,30]$ for which $\obs R(t)$ is halfway between $\obs R(0)=1$ and its average value $S_1(\obs R, \obs \Phi)^{1/2}$.
These statistics are connected to trajectories,$\phi_j(t)$ for $j =1,\dots,256$, of the high-fidelity model, \Cref{eq:Kuramoto_hi}, through the definition
\[
R(t) \exp (i \Phi(t)) = \frac{1}{256} \sum_{j=1}^{256} \exp(i \phi_j(t)).
\]
The low-fidelity model, in \Cref{eq:Kuramoto_lo}, directly models the evolution of $R$ and $\Phi$.
Example trajectories of the low-fidelity and high-fidelity model are shown in \Cref{fig:eg_dynamics}.

We can use the model in \Cref{eq:Kuramoto_lo} to justify the choice of summary statistics.
In particular, the steady-state value of $\tilde R$ is equal to
\[
\tilde R^\star = \left( 1 - 2\frac{\gamma}{K} \right)^{1/2},
\]
while we can write the solution $\tilde \Phi(t) = \omega_0 t$.
Then we use $S_1$ to approximate $(\tilde R^\star)^2 = 1 - 2\gamma/K$ and thus identify the ratio $\gamma/K$.
Similarly, $S_2 = \tilde \Phi(30)/30 = \omega_0$ allows us to directly identify $\omega_0$.
Finally, $S_3$ is a measure of the time-scale of the dynamics.
Trajectories with equal values for $S_1$ (i.e. equal steady states, and thus equal values for $\gamma/K$) can be distinguished by the speed at which they reach their steady state, which we will infer through $S_3$.
Note that the sampling point, $T_{1/2}$, used in $S_3$ is chosen to be relevant to the observed data specifically, and aims to distinguish any simulated trajectories from $\obs R$ and $\obs \Omega$ in particular.
Thus we select $S_3$ to identify the scale of $\gamma$ and $K$.
Hence, we assume that these three summary statistics will be sufficient to identify the parameters.

\section{Posterior samples}
\Cref{post:ABC:Rejection,post:ABC:SMC,post:ABC:SMC:ESS400,post:ABC:SMC:adaptive,post:MFABC:Rejection,post:MFABC:SMC:ESS400,post:MFABC:SMC:adaptive}
show samples from the posterior distributions $p_\epsilon((K, \omega_0, \gamma)~|~\obs y)$ approximating the Bayesian posteriors of the parameters for the Kuramoto oscillator network in \Cref{s:Example}.
The samples have been generated using \Cref{ABC:Importance} (ABC-RS), \Cref{ABC:SMC} (ABC-SMC), \Cref{MFABC:Importance} (MF-ABC-RS), and \Cref{MFABC:SMC} (MF-ABC-SMC).

The plots on the diagonal are one-dimensional empirical marginals for each parameter (i.e. weighted histograms).
The plots above the diagonal are all of the two-dimensional empirical marginals for each parameter pair, represented as heat maps.
The axes are discretised for this visualisation by partitioning each parameter's prior support into $B$ bins, where $B$ is the integer nearest to $(2 \times \mathrm{ESS})^{1/2}$.
For example, when $\mathrm{ESS} \approx 400$, each axis is partitioned into $28 \approx \sqrt{800}$ bins across its prior support.
The plots below the diagonal are all of the two-dimensional projections of the Monte Carlo set $\{ \theta_n, w_n \}$, where the weights $w_n$ are represented by colour.
In particular, negative weights are coloured orange and positive weights are purple.
Note that, for simplicity of visualisation, the weights are rescaled (without any loss of generality) to take values between $-1$ and $+1$.

\subsection{Existing ABC algorithms}
The posterior samples in \Cref{post:ABC:Rejection,post:ABC:SMC,post:MFABC:Rejection} from $p_{0.5}((K, \omega_0, \gamma)~|~\obs y)$ are generated by running \Cref{ABC:Importance,ABC:SMC,MFABC:Importance} for a fixed total of $N=6000$ parameter proposals and with threshold $\epsilon = 0.5$. 
For \Cref{ABC:SMC}, these proposals are split equally across four generations with decreasing thresholds $\epsilon_t = 2, 1.5, 1, 0.5$. 
The efficiency of generating these posterior samples is discussed in \Cref{s:existing}.

\subsection{Sequential Monte Carlo}
The posterior samples in \Cref{post:ABC:SMC:ESS400,post:MFABC:SMC:ESS400} from $p_{0.1}((K, \omega_0, \gamma)~|~\obs y)$ are generated by running the two SMC algorithms, \Cref{ABC:SMC,MFABC:SMC}, for eight generations with a common schedule of thresholds \[\epsilon_t = 2, 1.5, 1, 0.8, 0.6, 0.4, 0.2, 0.1,\] and with stopping condition of $\mathrm{ESS}=400$ at each generation.
Each figure is one representative output of the 50 runs of each of \Cref{ABC:SMC,MFABC:SMC} used in \Cref{s:Results}, where the efficiency of generating these posterior samples is discussed.

\subsection{Adaptive epsilon}
The posterior samples in \Cref{post:ABC:SMC:adaptive,post:MFABC:SMC:adaptive} are generated by the extension of the two SMC algorithms, \Cref{ABC:SMC,MFABC:SMC}, to allow for adaptive selection of thresholds $\epsilon$ as discussed in \Cref{s:Adaptive}.
Running the adaptive extensions of each of \Cref{ABC:SMC,MFABC:SMC} for four generations, with a fixed efficiency in each generation and stopping condition $\mathrm{ESS}=400$, produces the posterior samples in \Cref{post:ABC:SMC:adaptive,post:MFABC:SMC:adaptive}, respectively, from $p_{0.74}((K, \omega_0, \gamma)~|~\obs y)$ and $p_{0.1}((K, \omega_0, \gamma)~|~\obs y)$, respectively.

\begin{figure*}[p]
	\centering
	\includegraphics[width=\textwidth]{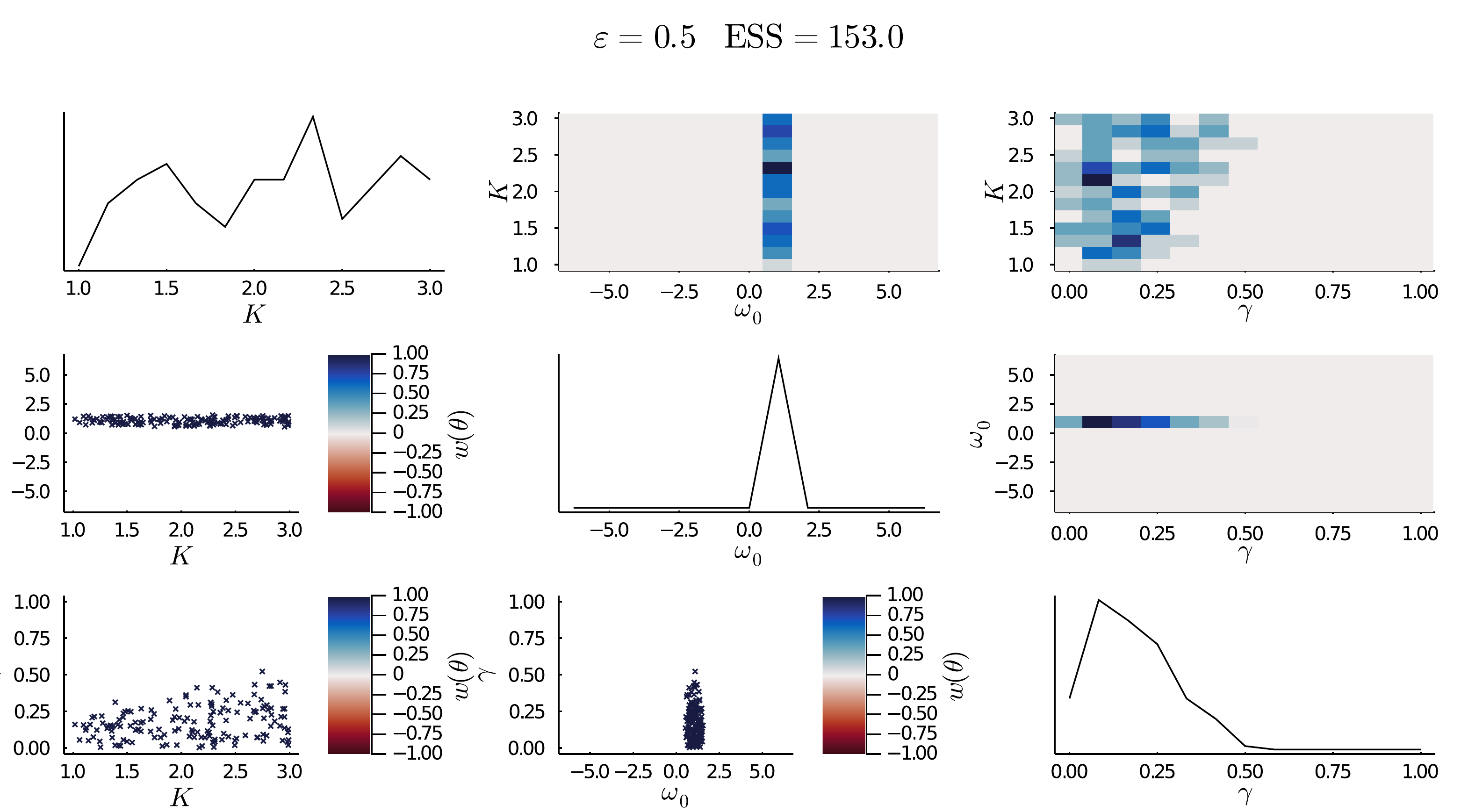}
	\caption{
		Sample from ABC posterior generated by \Cref{ABC:Importance}.
	}
	\label{post:ABC:Rejection}
\end{figure*}

\begin{figure*}[p]
	\centering
	\includegraphics[width=\textwidth]{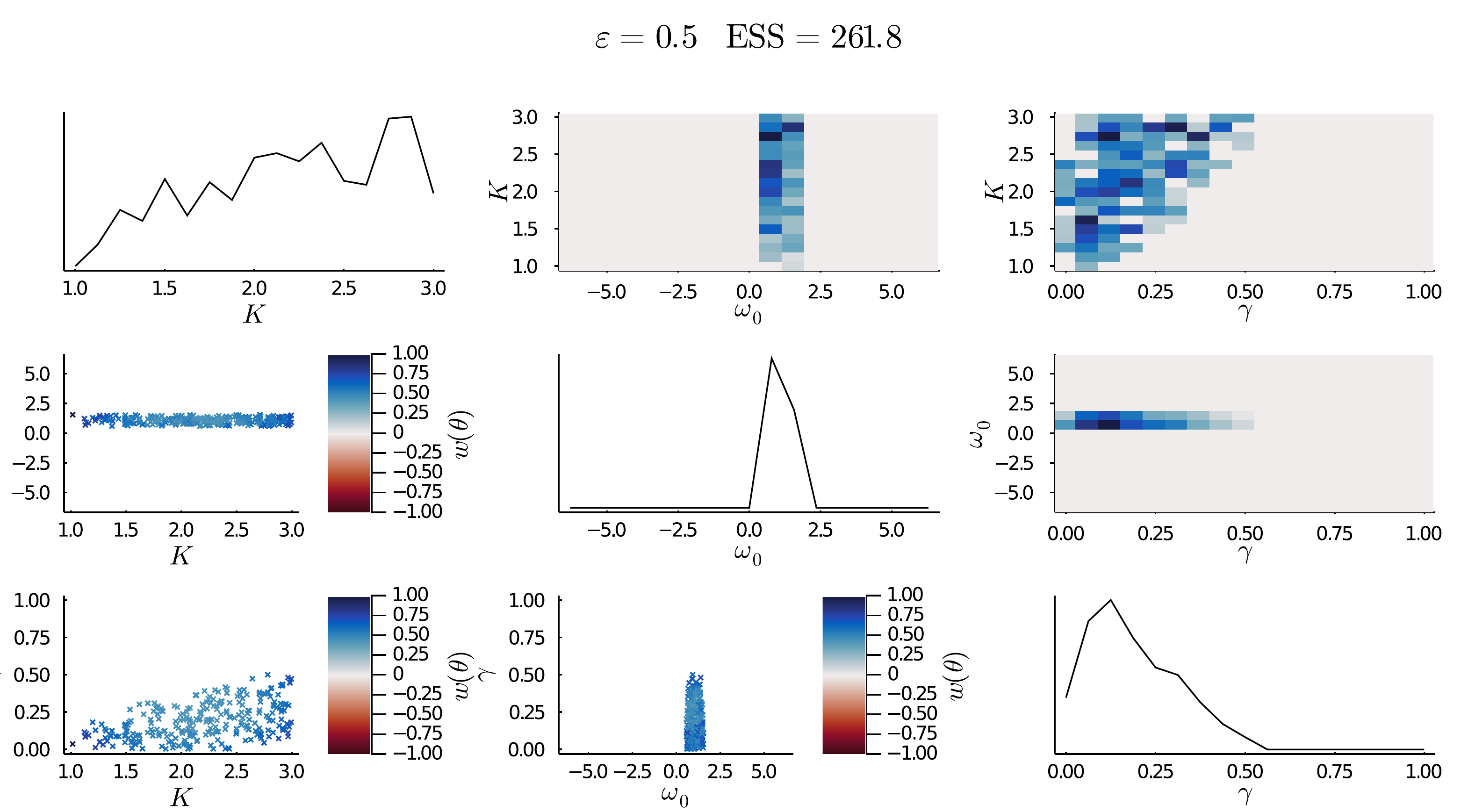}
	\caption{
		Sample from ABC posterior produced using the final generation of \Cref{ABC:SMC}.
	}
	\label{post:ABC:SMC}
\end{figure*}

\begin{figure*}[p]
	\centering
	\includegraphics[width=\textwidth]{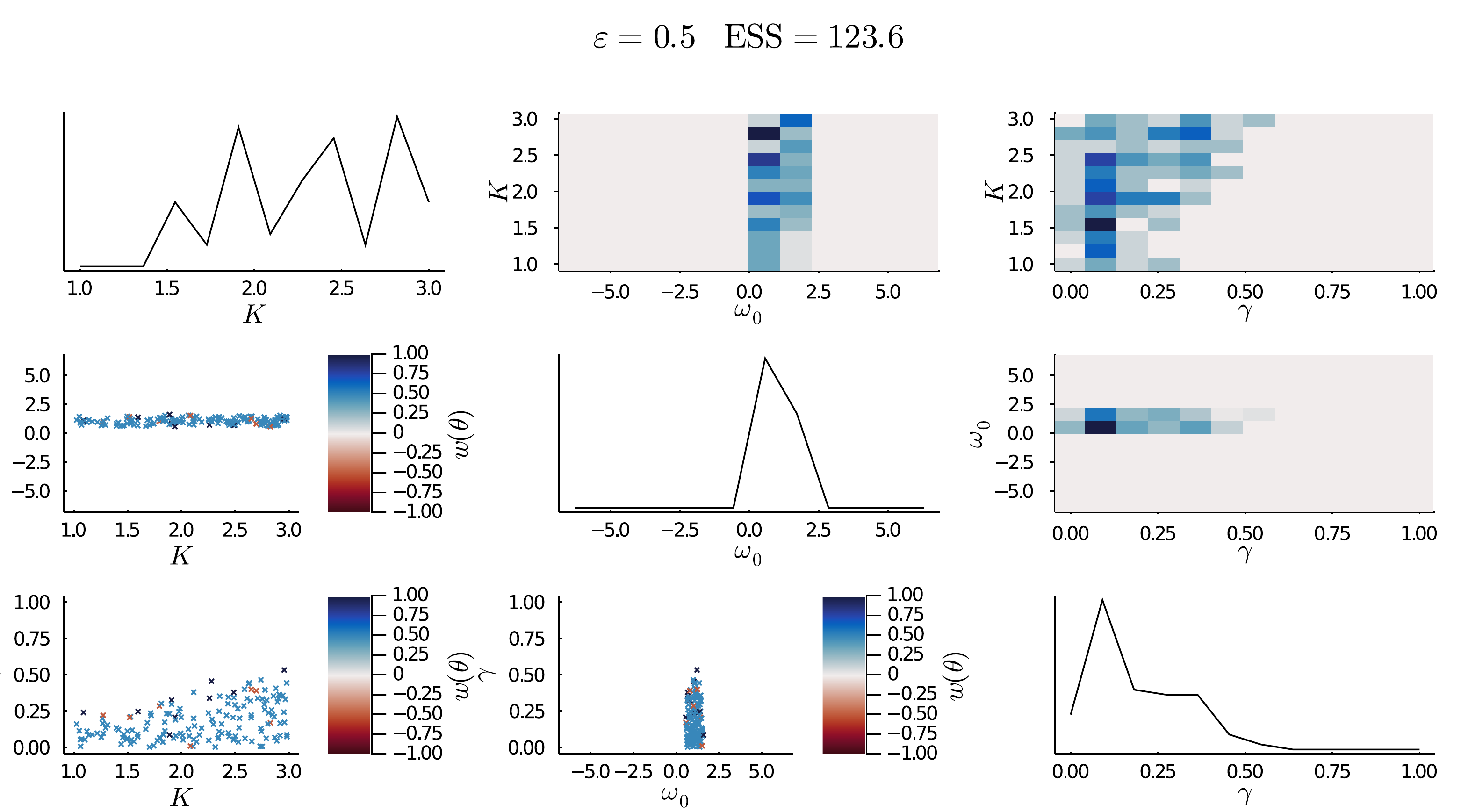}
	\caption{
		Sample from ABC posterior generated by \Cref{MFABC:Importance}.
	}
	\label{post:MFABC:Rejection}
\end{figure*}

\begin{figure*}[p]
	\centering
	\includegraphics[width=\textwidth]{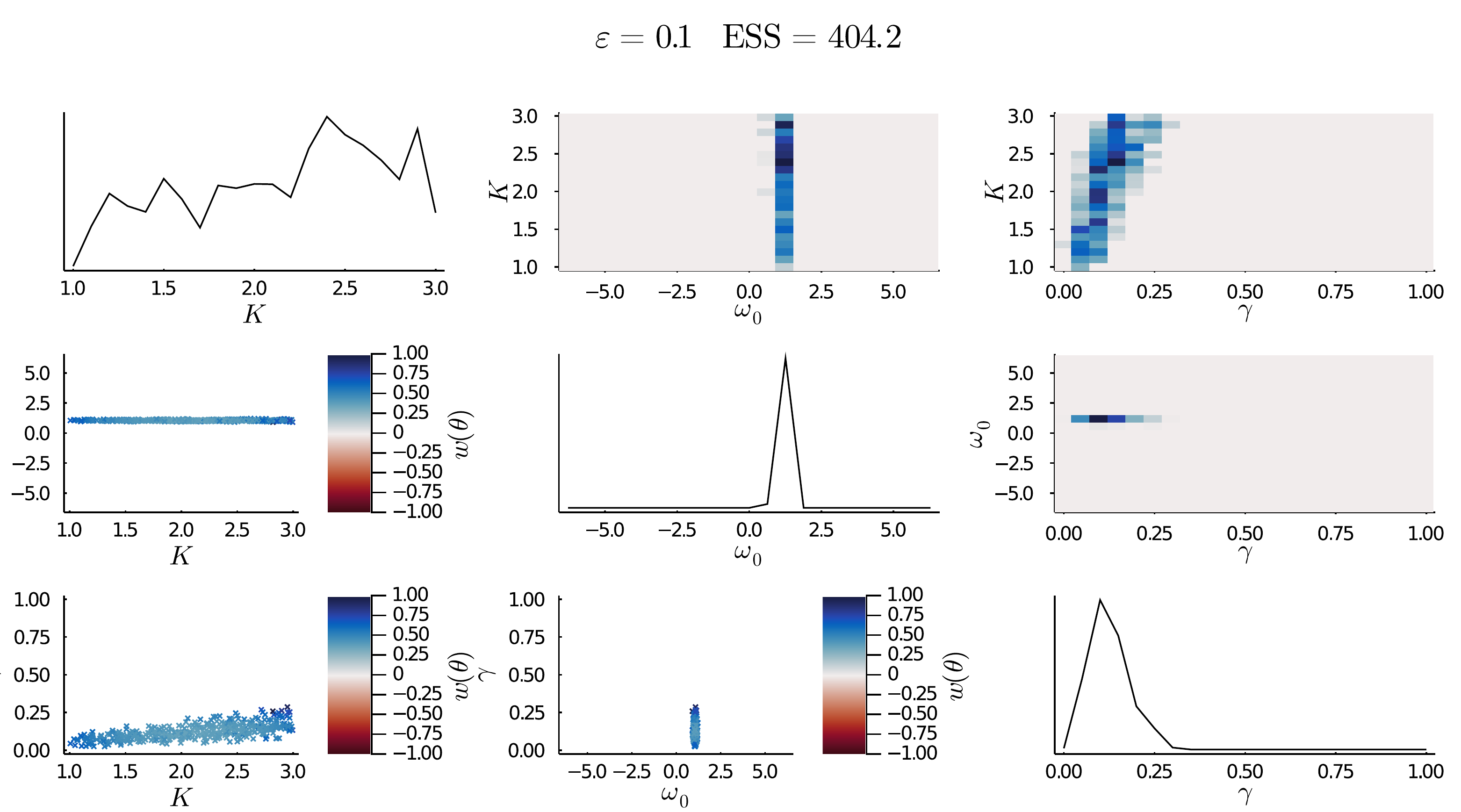}
	\caption{
		Sample from ABC posterior produced using the final generation of \Cref{ABC:SMC}.
	}
	\label{post:ABC:SMC:ESS400}
\end{figure*}

\begin{figure*}[p]
	\centering
	\includegraphics[width=\textwidth]{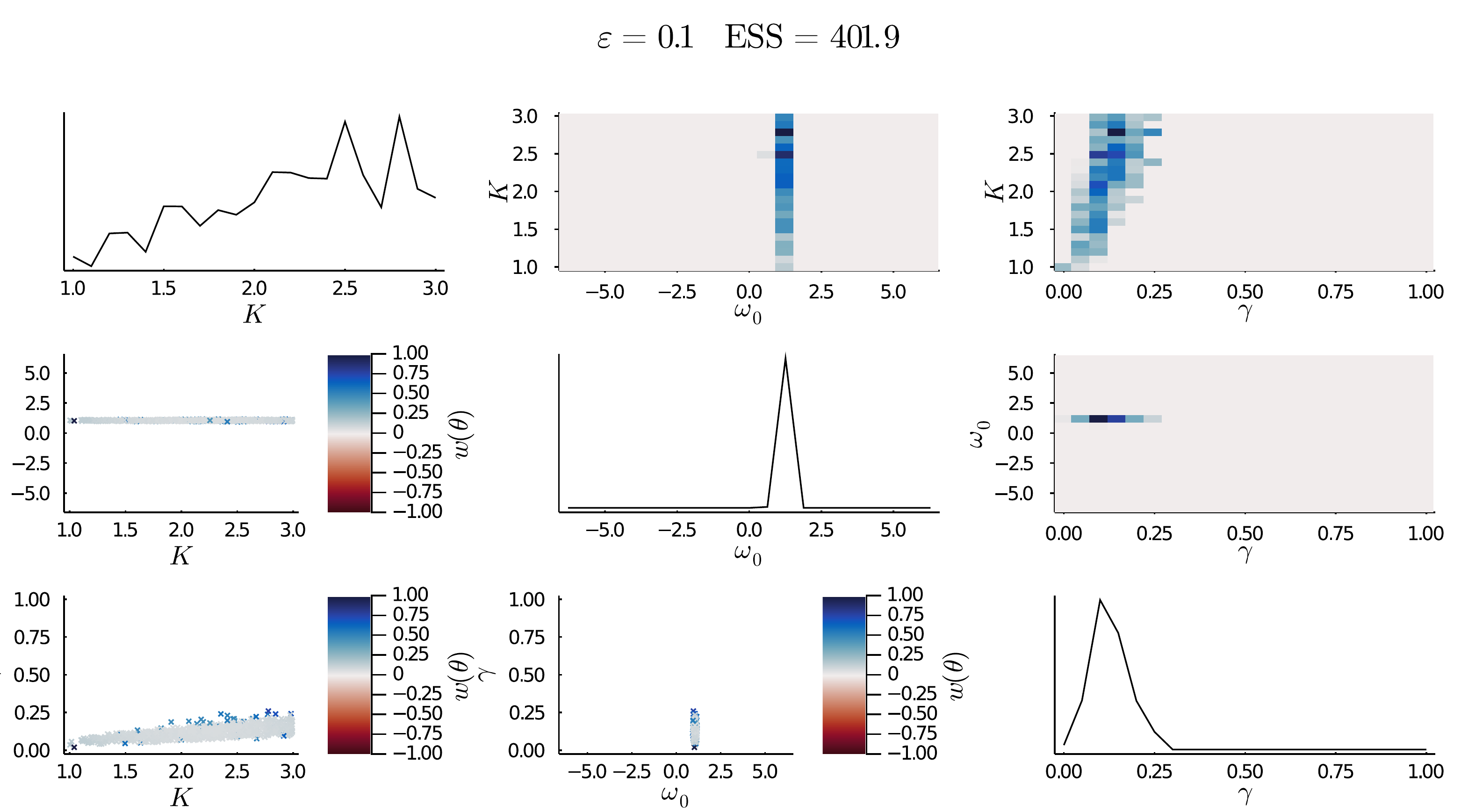}
	\caption{
		Sample from ABC posterior produced using the final generation of \Cref{MFABC:SMC}.
	}
	\label{post:MFABC:SMC:ESS400}
\end{figure*}

\begin{figure*}[p]
	\centering
	\includegraphics[width=\textwidth]{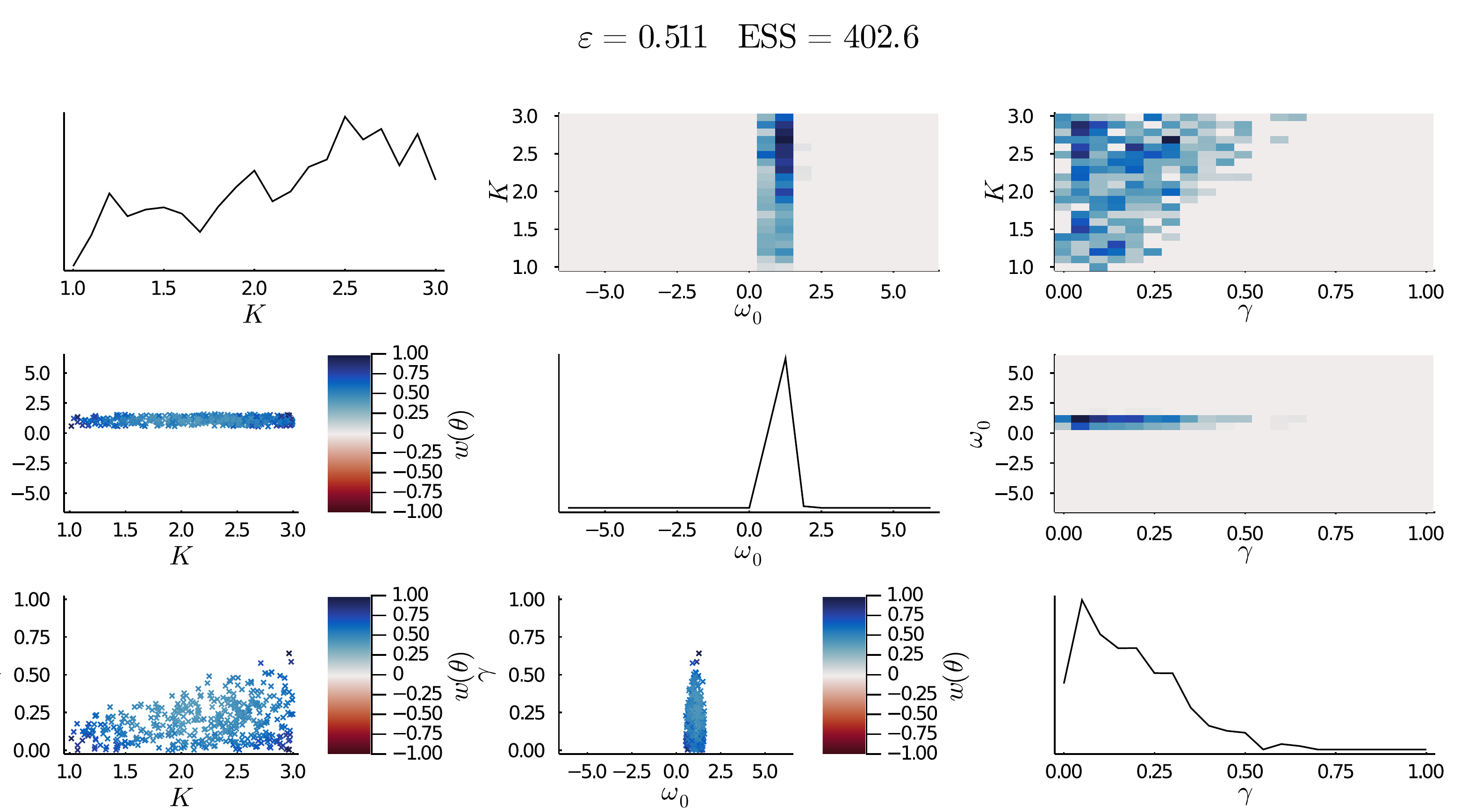}
	\caption{
		Sample from ABC posterior produced by the final generation of the adaptive modification of \Cref{ABC:SMC}, as described in \Cref{s:Adaptive}.
	}
	\label{post:ABC:SMC:adaptive}
\end{figure*}

\begin{figure*}[p]
	\centering
	\includegraphics[width=\textwidth]{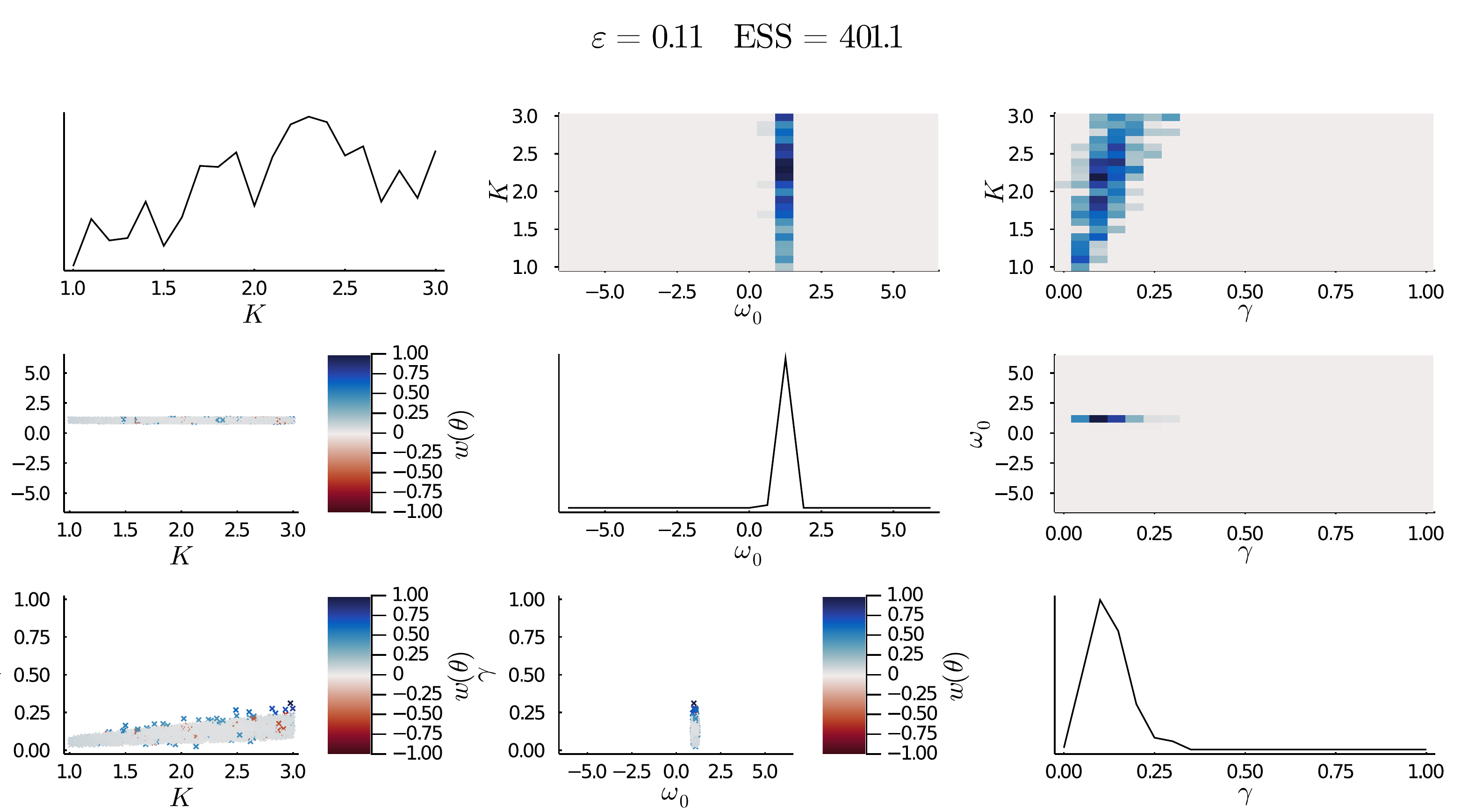}
	\caption{
		Sample from ABC posterior produced by the final generation of the adaptive modification of \Cref{MFABC:SMC}, as described in \Cref{s:Adaptive}.
	}
	\label{post:MFABC:SMC:adaptive}
\end{figure*}

\end{document}